\documentclass[ba]{imsart}
\pubyear{0000}
\volume{00}
\issue{0}
\doi{0000}
\firstpage{1}
\lastpage{1}

\usepackage{amsthm}
\usepackage{amsmath}
\usepackage{natbib}
\usepackage[colorlinks,citecolor=blue,urlcolor=blue,filecolor=blue,backref=page]{hyperref}
\usepackage{graphicx}
\usepackage{epstopdf}
\usepackage{rotating}

\startlocaldefs

\usepackage{placeins}
\usepackage[export]{adjustbox}
\RequirePackage[OT1]{fontenc}
\RequirePackage{amsthm,amsmath}

\numberwithin{equation}{section}
\theoremstyle{plain}
\newtheorem{thm}{Theorem}[section]

\newtheorem{lem}{Lemma}[section]

\newcommand{\R}{\mathbb R }

\def\1{1\!{\rm l}}
\usepackage{natbib}
\usepackage[colorlinks,citecolor=blue,urlcolor=blue,filecolor=blue,backref=page]{hyperref}

\usepackage{color}
\usepackage[export]{adjustbox}[2011/08/13]
\usepackage{array}
\usepackage{dsfont}
\usepackage{graphicx}
\usepackage{mathrsfs}
\usepackage{amssymb}
\usepackage{amsbsy}
\usepackage{amsfonts}
\usepackage{latexsym}
    \usepackage{caption}
\usepackage{subcaption}
\endlocaldefs

\begin{document}

\begin{frontmatter}
\title{Overfitting hidden Markov models with an unknown number of states}
\runtitle{Overfitting hidden Markov models}

\begin{aug}
\author{\fnms{Zo\'{e}} \snm{van Havre}\thanksref{addr1}\thanksref{addr2}\thanksref{e1} \ead[label=e1]{zoevanhavre@gmail.com}}, 
\author{\fnms{Judith} \snm{Rousseau}\thanksref{addr2}\ead[label=e2]{rousseau@ceremade.dauphine.fr}\ead[label=u1,url]{https://www.ceremade.dauphine.fr/}}, 
\author{\fnms{Nicole} \snm{White}\thanksref{addr1}},
\and
\author{\fnms{Kerrie} \snm{Mengersen}\thanksref{addr1}}

\runauthor{Z. van Havre et al.}

\address[addr1]{ARC Centre of Excellence for Mathematical and Statistical Frontiers (ACEMS), Queensland University of Technology (QUT), Australia
}

\address[addr2]{Centre De Recherche en Math\'{e}matiques de la D\'{e}cision (CEREMADE),  Universit\'{e} Paris-Dauphine, France 
}

\thankstext{t1}{Contact: \printead{e1} }

\end{aug}

\begin{abstract}
 This paper presents new theory and methodology for the Bayesian estimation of overfitted hidden Markov models, with finite state space. The goal is then to achieve posterior emptying of extra states. A prior configuration is constructed which favours configurations where the hidden Markov chain   remains ergodic although it empties out some of the states. Asymptotic posterior convergence rates are proven theoretically, and demonstrated with a large sample simulation. The problem of overfitted HMMs is then considered in the context of smaller sample sizes, and due to computational and mixing issues two alternative prior structures are studied, one commonly used in practice, and a mixture of the two priors.  The Prior Parallel Tempering approach of \citet{vanHavre2015} is also extended to HMMs to allow MCMC estimation of the complex posterior space. A replicate simulation study and an in-depth exploration is performed to compare the three priors with hyperparameters chosen according to the asymptotic constraints alongside less informative alternatives. 
\end{abstract}



\end{frontmatter}
 
\section{Introduction}\label{intro}
    Finite state space Hidden Markov Models (HMMs) arise when observations from a mixture of distributions depend on an unobserved (hidden) Markov chain. HMMs provide a framework for identifying and modelling homogeneous sub-sequences in data which display global heterogeneity, and as such are widely applicable in areas from DNA segmentation \citep{Churchill1989} to economic analyses \citep{Hamilton1989}.

In an HMM, the observed time series $y_{1:n}=\{ y_1, \dots, y_n\}$ depends on a single realisation of the underlying stochastic process determined by the unobserved states $x_{1:n}=\{x_1, \dots, x_n\}$:
 \begin{equation}
\label{model1 }
\forall t \leq n; \quad   [Y_t | X_t  = x] \sim g_{\gamma_{x}} \quad  \gamma \in \Gamma \subset \mathbb{R}^d, \quad X_t \in \mathbf{X}=\{ 1, \dots, K \}
\end{equation}
 where $(x_t)_{t\geq 1}$ is the realisation of a Markov chain with  $K$ states and transition matrix $Q=(q_{i,j})_{1 \leqslant i, j\leqslant K }$. 
The Markov chain is also associated with a stationary distribution which contains the long term state probabilities, $\mu_Q$ satisfying $\mu_Q Q = \mu_Q $.

Estimation of $Q$ and $(x_t)_{t \geq 1}$ is straightforward when the number of states ($K$) is known; detailed reviews can be found in \citet{Fruhwirth-Schnatter2006}, \citet{Cappe2005}, and \citet{Scott2002}. In a Bayesian context, MCMC estimation is particularly straightforward when conditionally conjugate priors are placed on the emission parameters, and this is combined with a data augmentation approach \citep{Chib1996, Tanner1987}. In such a set-up, the MCMC simply iterates between updating the transition probabilities given some estimate of the allocations ($Q|x_{1:n}, y_{1:n}$), as well as updating the emission parameters separately for each state ($\gamma_1, \cdots, \gamma_K|x_{1:n}, y_{1:n})$; this in turn leads to a new estimate of the posterior allocations  $(x_{1:n}| \gamma_i, i \leq K, Q, y_{1:n})$.

A more general setting is when the number of states $K$ is unknown and must be estimated. This problem of order estimation is notoriously difficult in the frequentist setting; \citet{Gassiat2000} show that the likelihood ratio statistic is unbounded even for the simple case of comparing models with $K=1$ and $K=2$ states. \citet{Chambaz2009} and \citet{Gassiat2003} proposed a solution based on implementing heavy penalties in a maximum likelihood setting, while \citet{Gassiat2002} employed penalised marginal pseudo-likelihood to obtain weakly consistent estimators for $K$.

In a Bayesian settings,  order estimation methods have included Reversible Jump Markov Chain Monte Carlo (RJMCMC) \citep{Richardson1997,Boys2004},
variational Bayes methods \citep{McGrory2009}, sequential inference methods \citep{Chopin2001}, Bayes factors \citep{Han2001,Friel2008}, and non-parametric methods \citep{Beal2002,Ding2010}.
\citet{Teh2006} propose a fully non-parametric framework in terms of both $K$ and the emission distributions, the Hierarchical Dirichlet
Process HMM (HDP-HMM), and adopt a model averaging approach over increasingly complex models to reach conclusions. Unfortunately, the HDP-HMM is not well suited for order estimation as  small variations in the data are known to cause this model to overestimate the number of groups as well as the frequency of transitions, resulting in a Markov chain which transitions extremely frequently between close, arbitrarily defined groups \citep{Fox2008}.
\citet{Fox2008} extend this approach to obtain a more sparse posterior, introducing the \textit{Sticky} HDP-HMM, where a small positive quantity is added to the prior on the diagonal of the transition matrix ($Q$) to increase the probability that the system will stay in its current state.

The main issue with order estimation stems from the non-identifiability which occurs when more states are included in a hidden Markov model than are supported by the observations. This is called \textit{overfitting} and it is an implicit aspect of any order estimation method, whether it must explore overfitted space via MCMC (such as RJMCMC), or fit at least one overfitted model as part of a comparison. Overfitting is becoming well understood in the case of finite mixtures, which can be considered simple, time-independent HMMs. A mixture model with $K^*$ groups can be always be explained equally well by one with $K$ groups, where $K>K^*$ the extra states are either empty or merged with the true groups, or some combination thereof \citep{Rousseau2011}. Asymptotic results by \citet{Rousseau2011} have proven that the prior on the mixture weights determines the posterior behaviour of extra groups, birthing a growing body of methods which directly use overfitting for order estimation in mixtures. Overfitting with the goal of \textit{emptying} extra groups for order estimation was further developed by \citet{vanHavre2015}, where the authors investigated priors which encourage extra groups to have posterior weights approaching zero. While the use of such priors would normally inhibit a Gibbs sampler from adequately exploring the posterior space, a parallel tempering approach on the prior terms was included which can obtain a well mixed sample from the desired target space.

In the case of finite mixture models, better control over the posterior configuration of extra states in overfitted models leads naturally to order estimation methods, and achieving posterior emptying of these states has several benefits. First, it results in a parsimonious description of the data. Second, the supported components are clearly differentiable from those deemed unnecessary by the model, as they are allocated no observations. Third, a single model is needed which only needs to be specified in a general form; for example, there is no need to design specific MCMC moves which create new components or take away components (known as birth, death, merge, and split moves), as this behaviour occurs naturally over the course of the sampler. Fourth, as the number of components is fixed and finite, the model is fully parametric and straightforward to estimate with relatively simple MCMC techniques \citep{vanHavre2015}.

While HMMs are closely related to finite mixture models, it is not as straightforward to obtain similar results for overfitted HMMs. In addition to the general non-identifiability issue, these models feature an additional complication: the hidden states are not independent of each other. 
When the number of states is known, \citet{Gunst2009} provide asymptotic consistency results for HMMs, but when this is not the case the literature is limited. \citet{Gassiat2012} have published the only result on asymptotic posterior convergence rates for overfitted, parametric, finite state space hidden Markov models. The authors find that the dependence between the states leads the neighbourhoods of the true parameter values to contain transition matrices which lead to non-ergodic Markov Chains, corresponding to areas of bad Markov behaviour.

While no theoretical result exists for inducing posterior emptying of overfitted finite state space HMMs, there are several indications in the literature that this is possible. Variational Bayes methods appear to depend on this through the so called ``state-removal phenomenon'', with no underlying theory \citep{McGrory2009}. In the case of the particle filter of \citep{Chopin2001}, the model considers only states which appear in the posterior sequence of hidden states, which is equivalent to ignoring empty groups. In the non-parametric setting, there is a growing body of literature in genetics where HMMs have been used for DNA segmentation. For example, \citet{Boys2000} and \citet{Nur2009} employ the familiar $\mathcal D( \alpha_{i1}, \cdots, \alpha_{iK})$  prior on row $i$ of the transition matrix $Q$, but specify an alternative structure to the standard $\alpha_{ij}= \alpha$ for all $i,j \leq K$. Instead, they assume it is unlikely one can detect short segments, except when searching for a state with known parameters or when there are many short segments from a particular state. The prior on row $j$ of the transition matrix is set to $\mathcal{D}( a,a,\cdots,a,d,a,\cdots, a )$, where $d$ is the $j$`th element and is larger than the exchangeable off-diagonal elements (so that $E(Q_{k,k})\rightarrow  1$). This is applied to HMMs with an unknown number of states in \citet{Boys2004}, as part of a Reversible Jump MCMC algorithm, which explores overfitted space. While this has been done with little supporting theory, it is interesting that this prior structure is seen repeatedly in this field. Since the results of these publications do not appear to contain large numbers of spurious components, there is some evidence they are successfully dealing with extra states by encouraging them to empty out.

The overall aim of this paper is to provide new theoretical results and develop methodology for overfitting HMMs with an unknown number of states, where extra states are allocated no observations \textit{a posteriori}. To this end, we consider three objectives. The first is to develop new asymptotic theory for overfitted hidden Markov models. In Section \ref{Section1}, new asymptotic results are provided for the posterior distribution of overfitted finite state space HMMs, and it is shown that the posterior of the stationary distribution of extra states can be confined to be arbitrarily small, yet remain ergodic, by implementing certain prior restrictions on the distribution of the transition probabilities. The second objective is to investigate the impact of the asymptotic constraints, as well as a more relaxed form of the prior, in a large sample context for univariate Gaussian HMMs. This is undertaken in Section \ref{Section2}, where a large sample simulation demonstrates that posterior emptying of extra states is possible in practice under certain prior constraints. 
Thirdly, the final objective is to explore the applicability of the asymptotic theory in the context Gaussian HMMs with small sample sizes, reflective of common HMM applications (in the order of hundreds). This is approached in two parts. First, in Section  \ref{Section3}, computational problems are identified and, as a consequence, three updated prior choices for overfitting Gaussian HMMs are proposed.  Mixing difficulties in the MCMC are addressed by extending the Prior Parallel Tempering algorithm of \citet{vanHavre2015} to the HMM case, allowing for simultaneous sampling of a family of overfitted posteriors.  Secondly, a simulation study is undertaken in Section \ref{Section4} to evaluate the proposed priors for overfitted Gaussian HMMs with small sample sizes.

\subsection{Notations and set-up} \label{sec:notations}

  We assume that the observations $y_ = (y_1, \cdots, y_n) $ are distributed according to model \eqref{model1 }  and we set $\mu$ as a prior initial distribution for the hidden Markov chain  $X_t=(x_1, \cdots, x_n)$. The conditional distribution of $y_t $ given $x_t = j$ is $G_{\gamma_j}$, which is absolutely continuous with respect to some fixed measure $\lambda$ with density $g_{\gamma_j}$.
  We set $\theta = (Q, \gamma_1, \cdots, \gamma_K) \in \mathcal Q_K\times \Gamma^K = \Theta_K$ where
  $$\mathcal Q_K = \{ Q= (q_{i,j})_{i,j\leq K}); \, \sum_{j=1}^K q_{i,j}=1, \, q_{i,j}\geq 0 \, \forall \, i,j \}.$$

  We also denote by $\mu_Q$ the (or one of the) stationary distribution associated to $Q$, i.e. the probability distribution on $\{1, \cdots, K\}$ satisfying
  $\mu_Q = \mu_Q Q $.

  We recall that for any Markov chain on a finite state-space with transition matrix $Q$, and stationary distribution $\mu_Q$ (one of them if $Q$ admits more than one stationary distribution), it is possible to define  $\rho_Q \geq 1$ such that for any $m$, any $i \leq K$
  \begin{equation*}
  \sum^k_{j=1} |(Q^m)_{ij} -\mu_{Q}(j)| \leq \rho^{-m}_Q, \quad \rho_Q=\left(1- \sum_{j=1}^K \min_{1\leq i\leq K}q_{i,j}  \right)^{-1}
  \end{equation*}

  The complete likelihood conditional on $X_1 =x_1$ is given by
  $$f_n(y_{1:n},x_{2:n}|\theta , x_1) = g_{\gamma_{x_1}}(y_1) \prod_{i=1}^{n-1} q_{x_i,x_{i+1}} g_{\gamma_{x_{i+1}}}(y_{i+1})$$
  and the likelihood conditional on $X_1 = x_1$, defining $x_{2:n}=(x_2, \cdots,x_n)$, is given by
  \begin{equation*} 
  f_n(y_{1:n}|\theta,x_1) = \sum_{x_{2:n}} f_{n}(y_{1:n},x_{2:n}|\theta,x_1).
  \end{equation*}
  
  We also write
   $$f_n(y_{1:n}|\theta,\mu) = \sum_{x_1=1}^K f_n(y_{1:n}|\theta,x_1)\mu(x_1)$$
   and $\ell_n(\theta, x_1) = \log f_n(y_{1:n}|\theta,x_1)$ and  $\ell_n(\theta, \mu) = \log f_n(y_{1:n}|\theta,\mu) $ .

  In this paper we study  the behaviour of  posterior distributions associated to priors belonging  to the following family :
  \begin{itemize}
  \item (C1) Prior on $Q$ : the rows $Q^i $ are independent and identically distributed according to a Dirichlet $\mathcal D(\alpha_1, \cdots, \alpha_K) $, with $\alpha_j>0$, $j \leq K$.
  \item  (C2)  Independent prior on the $\gamma$'s : $\gamma_j \stackrel{iid}{\sim} \pi_\gamma$ with positive and continuous density on $\Gamma$.
  \end{itemize}
  We denote by $\pi$ the above prior distribution and $\pi( . | y_{1:n})$ the corresponding posterior distribution, so that
  \begin{equation*}
  \pi( d \theta|  y_{1:n}) = \frac{ f_n(y_{1:n}|\theta,\mu) \pi(d\theta) }{ \int_\Theta f_n(y_{1:n}|\theta,\mu) \pi(d\theta) }.
  \end{equation*}

  We denote by $F(h) = \int h(x) dF(x) $ for every probability measure $F$ and integrable function $h$, also $\nabla f$  denotes the gradient of $f$ and $D^2 f$ its second derivatives.

  In the following section we study the asymptotic behaviour of the posterior distribution under some specific configurations of $(\alpha_1, \cdots, \alpha_K) $, when the true parameter corresponds to a HMM with $K^* < K$ hidden states. In the $K^* $- parameter space $\Theta_{K^*}$ we write $\theta^*$ the parameter, i.e.
  $\theta^* = (Q^*, \gamma_1^*, \cdots, \gamma_{K^*})$ with $Q^*\in \mathcal Q_{K^*} $ and $\gamma_j^* \in \Gamma$.
   The true underlying model can then be parametrized by infinitely many parameters in the $K$-parameter space $\Theta_K$. In particular any parameter of the form 
   $(\gamma_1^*, \cdots, \gamma^*_{K^*}, \gamma^*_{K^*}, \cdots , \gamma^*_{K^*}) \in \Gamma^K $
   and $Q$ with
   $q_{i,j} = q_{i,j}^* $ if $i \leq K^*$, $j \leq K^*-1$, $\sum_{j=K^*}^{ K}q_{i,j}= q_{i,K^*}$ and $q_{i,. }  = q_{K^*,.}$ for all $i\geq K^*+1$ leads to the same likelihood
  function $f_{n}(y_{1:n}|\theta^*,\mu) $, for all $\mu$. The parameters $\theta \in \Theta^K $ defined  by
  $$Q = \left( \begin{array}{cccc}
  Q^* & 0 & \cdots & 0\\
  R & 0 & \cdots & 0
  \end{array}\right) $$
  where for all $i =K^*+1, \cdots, K$, $R_{i,1}= 1$  and $R_{i,j}=0$  if $j \geq 2 $ and $\gamma_j = \gamma_j^*$ for all $j\leq K^*$ lead to the same likelihood function for all $\mu $ having support in $\{1, \cdots, K^*\}$. If $\mu(j)>0$ for some $j >K^*$, then the likelihood $f_{n}(y_{1:n}|\theta^*,\mu) $ is not the same however $f_n(y_{2:n} |\theta,x_2) = f_n(y_{2:n}|\theta^*,x_2) $ for all $x_2$, so that $f_n(y_{1:n}|\theta,\mu_Q)= f_n(y_{1:n}|\theta^*,\mu_{Q^*})$.

  We denote by  $ \Theta^* \subset \Theta_K$ the set of all $\theta$ such that either  $f_n(y_{1:n}|\theta,\mu_Q)= f_n(Y|\theta^*,\mu_{Q^*})$ or $f_n(y_{1:n}|\theta,\mu)= f_n(y_{1:n}|\theta^*,\mu)$  for all $n$.

\section{ Asymptotic behaviour of the posterior distribution } \label{Section1}
    In \citet{Gassiat2012}, posterior concentration for HMMs models \eqref{model1 } has been obtained in terms the $L_1$ distance between the stationary marginal distributions of two consecutive observations and an estimator of the number of components has been proposed. In the special case where $K=2$ the authors also prove that the posterior distribution concentrates on the configuration where the extra component merges with the \textit{true} one, under some conditions on the prior of the transition matrix. In this section, we are interested in finding some sufficient conditions on the prior to ensure that the posterior distribution concentrates on the configuration where the extra states are emptied out when the number of observations goes to infinity.

    We consider the following regularity conditions on $g_\gamma$, resembling those of \citet{Rousseau2011} and \citet{Gassiat2012}.

    \begin{itemize}
    \item [A1]  \textit{Regularity}: The model $\gamma \in \Gamma \rightarrow g_\gamma$ is twice continuously differentiable and regular in the sense that
     for all $\gamma \in \Gamma $ the Fisher information matrix associated with the model $g_\gamma$ is positive definite at $\gamma$.
    For all $i \leq K^*$, there exists $\delta>0$ such that for all
     $$G_{\gamma_j^*} \left( \sup_{|\gamma-\gamma_j^*| <\delta }  |\nabla_\gamma \log g_{\gamma} |  \right) <+\infty, \quad
    G_{\gamma_j^*}  \left(  \sup_{|\gamma-\gamma_i^*|\leq \delta} |D^2_\gamma \log  g_{\gamma}| \right) <+\infty, $$
    and for $r \in \{1,2\}$,
    $$\int  \sup_{|\gamma - \gamma_j^*|<\delta}\left| D^r g_{\gamma}(y) \right|   d\lambda(y) < +\infty. $$

    Assume also that for all $i=1,...,K^*$ $\gamma_{i}^* \in \mbox{int}(\Gamma)$ the interior of $\Gamma$.

     \item [A2] \textit{Offset}: There exists an open subset $\Gamma_ 0 \subset \Gamma$ satisfying $\mbox{Leb}(\Gamma_0) >0$,
     and for all $i \leq K^*$
    $$d(\gamma_i^0 , \Gamma_0) = \inf_{\gamma \in \Gamma_0} |\gamma - \gamma_i^0| > 0$$
     and such that, there exists $\delta>0$ such that
    $$ \left\| \frac{  \sup_{\gamma \in \Gamma_0} g_\gamma(.) }{ \max_{i\leq K^*} \inf_{|\gamma' - \gamma_i|<\delta} g_{\gamma'}(.)} \right\|_\infty<  +\infty. $$

     \item [A3] \textit{Stronger identifiability}:  For any  $\mathbf{t}=(t_{1},\ldots,t_{k_{0}})\in T$,
    any $(\pi_i)_{i = 1}^{k - t_{K^*}} \in (\R^+)^{ k - t_{K^*}}$ (if $t_{K^*} <k$), any $(a_i)_{i=1}^{K^*}, (c_i)_{i=1}^{K^*} \in \R^{K^*}$, $(b_i)_{i=1}^{K^*}\in (\R^d)^{K^*}$, any $z_{i,j} \in \R^d,\; \alpha_{i,j}\in\R,\; i=1,\ldots,K^*, j=1,\ldots, t_i - t_{i-1}$, with $t_0= 0$ such that $\| z_{i,j}\| = 1$, $\alpha_{i,j}\geq 0$ and $\sum_{j=1}^{t_i - t_{i-1}}\alpha_{i,j}=1$,
     for any $(\gamma_i)_{i=1}^{k - t_{K^*}}$ which belong to $\Gamma \setminus \{ \gamma_i^*, i=1,\ldots,K^*\}$,
    \begin{eqnarray} \label{iden:1}
    \sum_{i=1}^{k - t_{K^*}} \pi_i g_{\gamma_i} + \sum_{i=1}^{K^*} \left( a_i g_{\gamma_i^*} + b_{i}^{T}D^1g_{\gamma_i^*}\right) + \sum_{i=1}^{K^*}  c_i^2\sum_{j=1}^{t_{i}-t_{i-1}} \alpha_{i,j} z_{i,j}^{T}D^2 g_{\gamma_i^*} z_{i,j} = 0,
    \end{eqnarray}
     if and only if $$ a_i = 0, b_i = 0, c_{i}=0   \quad \forall i=1,\ldots,K^*,   \quad \pi_i=0 \quad \forall i = 1,\ldots,k-t_{K^*}.$$

    \end{itemize}

    \begin{thm}\label{th:postempty}
    Consider the model \eqref{model1 } with prior defined by conditions (C1) and (C2). Assume that the true model is a HMM on a $K^*<K$ hidden states with true parameter $\theta^* \in \Theta_{K^*}$ defined by  $Q^* = (q_{i,j}^*) \in \mathcal Q_{K^*}$ with $q_{i,j}^*>0$ and $(\gamma_1^*, \cdots, \gamma_{K^*}^*)$ and that the regularity conditions [A1]-[A3] are satisfied.
    If there exists $1 \leq p \leq K^* $ such that $\alpha_1= \cdots = \alpha_p = \bar \alpha $ and $\alpha_{p+1} = \cdots = \alpha_K = \underline \alpha $ satisfying

{\footnotesize
    \begin{equation}\label{cond:alpha}
    \begin{split}
    &  p\bar \alpha + (K-p) \underline \alpha > \frac{ (K^*(K^*-1+d)+ \underline \alpha K( K-K^*))(K^*(d+K^*-1)+ \underline{\alpha}(K^*+1)(K-K^*-1)+ d/2)  }{ d/2 - \underline \alpha[ (K-K^*)^2 -(K-2K^*-1)]} \\
    & d/2 > \underline \alpha ((K-K^*)^2 - (K-2K^*-1)) \nonumber
    \end{split}
    \end{equation}
}

    then setting $A_1= K(K-K^*)\underline \alpha + K^*( K^*-1+d)$ and $A = A_1/ (p\bar \alpha + (K-p)\underline{\alpha} )$, for any $M_n$ going to infinity,
    \begin{equation}\label{post:th:1}
    \begin{split}
     &\pi( \min_{\sigma\in \mathcal S_K}  \sum_{j=K^*+1}^Kp_{\sigma(j)} > M_n v_n | y_{1:n}) = o_p(1), \text{ and} \\
     &v_n = n^{-1/2[( 1-A)B- A_1]/(d/2 + \underline \alpha (K-2K^*-1))}(\log n)^{B/( d + 2 \underline \alpha (K-2K^*-1))}
    \end{split}
    \end{equation}
    with $B = K^*(d+K^*-1)+ \underline{\alpha}(K^*+1)(K-K^*-1)+ d/2$.
    \end{thm}

    Since $K_0$ is unknown, $\bar \alpha$ and $\underline \alpha$ have to be chosen in a conservative way, which corresponds, given the above equations to $K^* = K-1$ for the lower bound on $\bar \alpha$ and $K^*=1$ for the upper bound on $\underline \alpha$, so that \eqref{cond:alpha} becomes
    \begin{equation}\label{cond:alphaConservative}
    \begin{split}
    & d/2 > \underline \alpha (K^2 - 3K+4)\\
    & p\bar \alpha + (K-p)\underline \alpha > \frac{( (K-1)(K-2+d)+ \underline \alpha K)((K-1)(d+K-2)+ d/2)  }{ d/2 - \underline \alpha[K^2 - 3K+4]}
    \end{split}
    \end{equation}

    Assumptions \textbf{A1}-\textbf{A2} are very similar to those found in \citet{Rousseau2011}. The offset condition is slightly stronger since instead of a moment condition on $g_\gamma/g_{\gamma_i}$ for $\gamma \in \Gamma_0$, and $i\leq K^*$, it is required to be bounded however it is satisfied in many parametric families, by choosing appropriately $\Gamma_0$. For instance in the case of location - scale Gaussian distributions, with $\gamma = (\mu, \sigma)$ and $\gamma_i = (\mu_i,\sigma_i)$
    $$\frac{ g_\gamma}{ g_{\gamma_i}}(x) \propto e^{- \frac{x^2}{2}( \sigma^{-2}- \sigma_i^{-2})  + x( \mu\sigma^{-2}- \mu_i\sigma_i^{-2})}$$
    is bounded as soon as $\sigma< \sigma_i$. More generally for any exponential family with density in the form $h(x)e^{\gamma^t T(X) - \psi(\gamma)}$ assumption \textbf{A2} is satisfied as soon as for any fixed set $(\gamma_1, \cdots, \gamma_k)$ there exists an open set $\Gamma_0$ such that  for all $i$
     $\sup_{\gamma \in \Gamma_0} \sup_x (\gamma - \gamma_i)^t T(x) <+\infty$. This is satisfied in most cases, as in the cases of Gamma, Poisson, and others of the exponential family, for example. Note also that assumption \textbf{A2} is also verified in the case of student distributions with $\gamma$ possibly covering the location, scale and degrees of freedom parameters.

    The proof of Theorem \ref{th:postempty} is provided in the Supplementary Material (Appendix A), however we sketch below its main arguments.

    To prove Theorem \ref{th:postempty},  we build on Theorem 2 of \citet{Gassiat2012}, which states that under the above  conditions
    $$\pi \left( \|f_{2}(. | \theta,\mu_Q) - f_2(. |\theta^*,\mu_{Q^*})\|_1 (  \rho(Q) -1 ) \leq \sqrt{\log n}/\sqrt{n} \right) = 1 + o_p(1) .$$

    To eliminate the term $  \rho(Q) -1$, we decompose, for any measurable set $B \subset \Theta_K$,
    \begin{equation}
    \pi( B | Y )  = \frac{ \int_B e^{\ell_n(\theta, \mu) - \ell_n(\theta^*, \mu) }\pi(d\theta) }{ \int_\Theta e^{\ell_n(\theta, \mu) - \ell_n(\theta^*, \mu) }\pi(d\theta) } := \frac{ N_n(B) }{ D_n}
    \end{equation}
    and we derive a sharper lower bound of $D_n$ than in \citet{Gassiat2012} who obtained  $D_n \gtrsim n^{-K(K-1+d)}$, with probability going to 1.

    To lower bound $D_n$, we approximate $\ell_n(\theta^*,x_1)$ by $\ell_n( \theta, x_1)$ with $x_1 = 1$ and $\theta \in S_n$ with
    \begin{equation}\label{Sn}
    \begin{split}
    S_n &= \left\{ \theta = (Q, \gamma_1, \cdots, \gamma_K); |q_{i,j}- q_{i,j}^*| \leq 1 /\sqrt{n} ,\,  i,j \leq K^* ; \, q_{i,j} \in (1/(2\sqrt{n}), 1/\sqrt{n}), \, i,j >K^*,\right. \\
     & \qquad \left.  q_{i,j} \in q^o_{i,j}\pm e_n, i>K^*, j\leq K^*, |\gamma_j-\gamma_j^*|\leq 1/\sqrt{n} , |\gamma_j - \gamma_j^o | \leq e_n\right\}
     \end{split}
    \end{equation}
    where $e_n= o(1) $,  $q_{i,j}^o >0$ and satisfies $\sum_{j=1}^{K^* } q_{i,j}^o= 1$ for all $i>K^* $ and $\gamma_j^o \in \Gamma_0$ for $j>K^*$.
    We have
    \begin{equation}\label{piSn}
    \pi(S_n)  \gtrsim (n^{-1/2})^{K^* d + K(K-K^*)\underline \alpha + K^*(K^*-1) }e_n^{(K-K^*)(d+K^*-1)}
    \end{equation}
    and we show in Lemma A.1 (see Appendix A) that for all $\epsilon >0$ there exists $C_\epsilon>0$ such that   for all $\theta \in S_n$,
    \begin{equation*}
    \mathbb P \left( \ell_n(\theta, x_1=1) - \ell_n(\theta^*, x_1=1) \leq -C_\epsilon \right) < \epsilon.
    \end{equation*}
    This ensures that as soon as $\mu(1)>0$,
     $$ D_n \gtrsim n^{- [K^*( K^*-1+d) + \underline \alpha K (K-K^*) ]/2} e_n'$$
     for any $e_n' = o(1)$, with probability going to 1.

    The rest of the proof, then follows the same lines as in \citet{Gassiat2012}. We define $\mathcal Q_n \subset \mathcal Q$, by
    $$\mathcal Q_n = \{ Q\in \mathcal Q; \sum_j \min_i q_{i,j} > v_n \} $$
    then $\pi( \mathcal Q_n^c) = O(v_n^{\sum_i \alpha_i}) = O(v_n^{p\bar \alpha +(K-p)\underline \alpha })$
    and writing $v_n = n^{-1/2} w_n$ with $w_n \rightarrow +\infty $,  if
    \begin{equation}\label{contrainte0}
    w_n^{(p\bar \alpha+ (K-p)\underline \alpha)/2}  = o( (\sqrt{n})^{p\bar \alpha- \underline \alpha (  K(K-K^*-1) +p) - K^*( K^*-1+d) })
    \end{equation}
    then
    $$\pi( \mathcal Q_n^c | y_{1:n})= o_p(1) .$$
     Note that \eqref{contrainte0} implies \eqref{cond:alpha}.
    Under \eqref{cond:alpha},  then defining $A_n = \{ \|f_2(.|\theta^*,\mu^*) - f_2( .|\theta \mu) \|_1 \lesssim \sqrt{\log n} /w_n\}$
     $$\pi( A_n |y_{1:n} )= 1 + o_p(1), $$
    with $w_n$ satisfying \eqref{contrainte0}.

    The proof is achieved by bounding from above
     $ \pi ( B_n) $, where
     \begin{equation}\label{Bn}
     B_n = \{ \theta \in A_n; \sum_{j= K^*+1}^K \mu_Q(j) > M_nu_n \},
     \end{equation}
    which is done in Lemma A.2 (see Appendix A).

Theorem \ref{th:postempty} means that it is possible to achieve posterior emptying of extra states in HMMs by binding the posterior distribution of the $\mu_k$ associated with extra states to be small yet remain non-zero, thereby retaining the ergodicity of the estimated Markov chain. This behaviour is possible due to the structure of the prior, which is asymmetrical with respect to the hyperparameter values. It contains large values $\bar{\alpha}$ and smaller values of $\underline{\alpha}$, for which asymptotic constraints are provided and interpreted to produce conservative guidelines given in Equation \ref{cond:alphaConservative}. Note that these are probably not sharp, due to the complex nature of the overfitted posterior parameter space in HMMs.
The prior constraints depend on $d$, the number of free parameters in each state, and the total number of states in the model $K$, but also on $p$, which defines the number of $\bar{\alpha}$ values included in the prior. As $p>K^*$ must hold, $p$ must be set to the smallest reasonable value; in our simulations we have set $p=1$, corresponding to a noninformative set-up on the number of components.

\section{Verification: large sample simulation}\label{Section2}

A large sample simulation experiment is performed to demonstrate the newly discovered theoretical results.

\subsection{Methodology}
For the illustrative portion of this paper, we consider HMMs with normally distributed state specific distributions, 
\[ [Y_t|X_t=j] \sim \mathcal{N}(\gamma_{j}, 1) \quad X_t \in \mathbf{X}=\{ 1, \dots, K \}\]

A Gibbs sampler is set up on the augmented parameter space $p(x_{1:n}, \gamma, Q | y_{1:n})$ (See Eq.\ref{eq:HMMaugParSpace})  as described by \citep{Fruhwirth-Schnatter2006}.
The prior on the emission means is set to $\pi(\gamma)\sim \mathcal{N}(\gamma_0=\bar{y}, \tau_0=100)$, where $\bar{y}$ is the observed sample mean. The hyperparameter of the prior on the mean, $\gamma_0$, is set to the global mean, and the prior variance of all means set to 100. The prior on each row of $Q$ follows a Dirichlet distribution of the form $\mathcal{D}(\bar{\alpha}, \underline{\alpha}, \cdots, \underline{\alpha})$.  The sampler is run for $M=20,000$ iterations and the first $10,000$ are discarded. The details of this algorithm can be found in the Supplementary Material (Appendix B).

\begin{equation}\label{eq:HMMaugParSpace}
\pi(x_{1:n}, \gamma, Q | y_{1:n}) \propto p_n(y_{1:n}|x_{1:n},  \gamma ,Q )p(x_{1:n}|\gamma ,Q ) \pi(\gamma) \pi(Q)
\end{equation}


A sample of $10,000$ observations is simulated from a univariate Gaussian HMM with $K^*=2$ states, transition matrix
$Q^*=\left\{ \begin{array}{ll}
        0.6 & 0.4 \\
        0.7 & 0.3 \end{array} \right\}$,
leading to stationary distribution $\mu^*= \{0.64, 0.36\}$. The individual states are distributed according to a univariate Gaussian distribution with emissions means $\gamma^*=\{-1, 3\}$, and variances assumed to be known and equal to one.

This simulation is modelled with $K=4$ states to explore the effect of the prior. Nine combinations of hyperparameters are explored, using the values $\bar{\alpha}=Theory, \quad K,\text{ and }1$, and $\underline{\alpha}= \bar{\alpha}, \quad 0.01, \text{ and } \frac{1}{n}$.

\subsection{Results}
We refer to the number of occupied states at iteration $(m)$ as $K_A^{(m)}$, where $\hat{K_A}$ is its empirical mode (over $m=1,\cdots,M$). The distribution of the set $(K_A^{(1)}, \cdots,K_A^{(M)})$ is referred to as $P(K_A)$.

The structure of the prior on $Q$ resulted in both merging and emptying of extra states \emph{a posteriori} depending on the choice of $\underline{\alpha}$. Table \ref{HMM_Table_1} details $P(K_A)$ for each prior combination, and Figure \ref{HMM_Figure_1} includes a visual representation of the posterior distribution estimated by each model using two-dimensional density plots of $\mu_k\times \gamma_k$ for all states $k=1,\cdots, 4$.

A choice of $\underline{\alpha}=\bar{\alpha}$ caused all states to be occupied regardless of the value of $\bar{\alpha}$, and $P(K_0=4)=1$ in all cases. The upper row of plots in Figure \ref{HMM_Figure_1} show while all $\mu_k$ were non-negligible in this case, the posterior space created by each different value of $\bar{\alpha}$ was noticeably different. For $\bar{\alpha}=1$, extra states merged to some degree with the true states, and the number of modes and their means visibly reflected the two true states. A modest increase in $\bar{\alpha}$ (to$\bar{\alpha}=K=4$) did not cause a large change in this instance, but the right-hand mode became less distinct and began to split into two. For the largest value of $\bar{\alpha}$, three separate modes were created as the extra states did not merge with the truth but merged together to create a new spurious mode not close to true parameter values.

The use of an asymmetrical prior where $\bar{\alpha}>\underline{\alpha}$ caused extra states to be assigned small stationary distributions and resulted in few or no observations to be allocated to the extra states during the MCMC. This behaviour was most clearly observed under the theoretically given constraints, where $P(K_0=2)=1$ and 0.97 given $\underline{\alpha}=\frac{1}{n}$ and 0.01 respectively. The lowest value of $\underline{\alpha}$ also induced a similar degree of posterior emptying when combined with $\bar{\alpha}$ which are not as extreme as this value, however. This behaviour softened for $\underline{\alpha}=0.01$; $P(K_0=3)=0.14$ for $\bar{\alpha}=4$, and 0.11 for $\bar{\alpha}=1$, with some iterations containing 4 occupied states.

The plots in the lower row in Figure \ref{HMM_Figure_1} show that for all values of $\bar{\alpha}$, a similar posterior space was created when $\underline{\alpha}=\frac{1}{n}$. Two clearly defined modes were positioned at the true parameter values, atop a `pool' of samples representing MCMC draws from empty states, where the $\gamma_k$ were drawn directly from the prior.

For all the asymmetrical priors explored where $\bar{\alpha}>\underline{\alpha}$, the separation evident between the posterior modes resulted in a total lack of mixing in the MCMC. No label switching was observed once the burn-in period has passed as the extreme posterior surface prevents the sampler from exploring the full posterior space. This will be investigated as part of Section \ref{Section3}.

\begin{table}[htb]
\caption[Large sample demonstration: proportion of occupied states]{Proportion of iterations representing each number of occupied states, obtained from fitting a HMM with $K=4$ states to $n=10,000$ observations from a simulated HMM with $K^*=2$ true states. The values of $\bar{\alpha}$ and $\underline{\alpha}$ refer to the hyper-parameters of the prior on each row of the transition matrix, of the form   $q_{i,.}\sim \mathcal{D}(\bar{\alpha}_{1},  \underline{\alpha}_{2}, \underline{ \alpha}_{3},  \underline{\alpha}_{4})$}
\centering
\begin{tabular}{r|r|rrr}
\hline
$\bar{\alpha}$  & $\underline{\alpha}$  & $P(K_A=2) $           & $P(K_A=3) $             & $P(K_A=4) $ \\ \hline
172  (Theory)         & $\bar{\alpha}$        & 0.0000                & 0.0000                  & 1.0000 \\
~               & 0.01                  & 0.9651                & 0.0348                  & 0.0000 \\
~               & 1/n        & 1.0000                & 0.0000                  & 0.0000 \\ \hline
4 (K)           & $\bar{\alpha}$        & 0.0000                & 0.0000                  & 1.0000\\
~               & 0.01                  & 0.8539                & 0.1446                  & 0.0015  \\
~               & 1/n        & 1.0000                & 0.0000                  & 0.0000  \\ \hline
1               & $\bar{\alpha}$    & 0.0000  & 0.0000    & 1.0000 \\
~               & 0.01                  & 0.8782                & 0.1130                  & 0.0088 \\ \hline
~               & 1/n       & 1.0000                & 0.0000                  & 0.0000 \\
\end{tabular}
\label{HMM_Table_1}
\end{table}

\begin{figure}[htb]
  \centering
  \begin{subfigure}{0.32\textwidth}
    \adjincludegraphics[trim={0 0 0 {.5\height}},clip,width=1\linewidth]{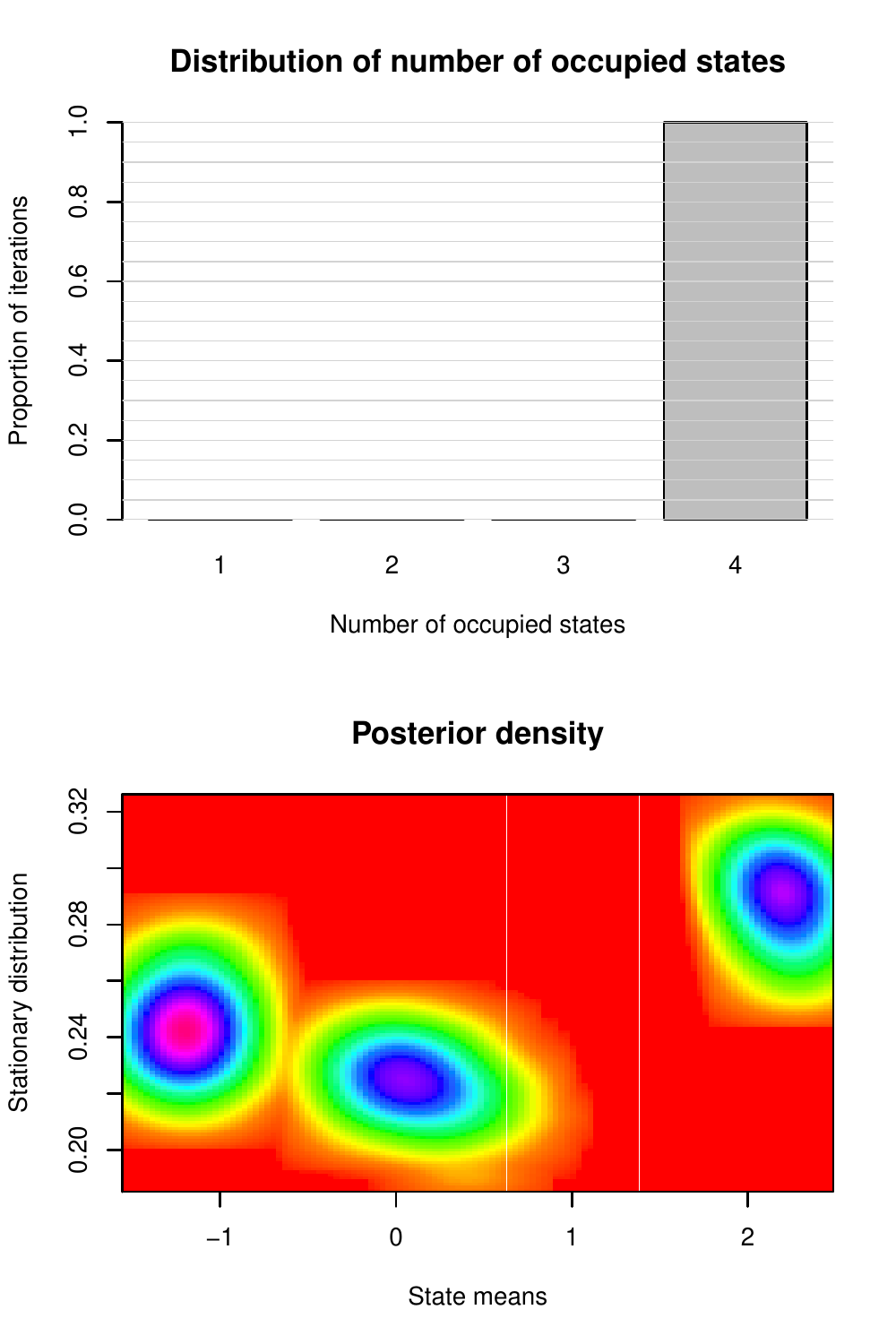}
\caption{$q_{i,.}\sim \mathcal{D}(172, 172, 172, 172)$}
\label{fig:sub1}
\end{subfigure}
  \begin{subfigure}{0.32\textwidth}
        \adjincludegraphics[trim={0 0 0 {.5\height}},clip,width=1\linewidth]{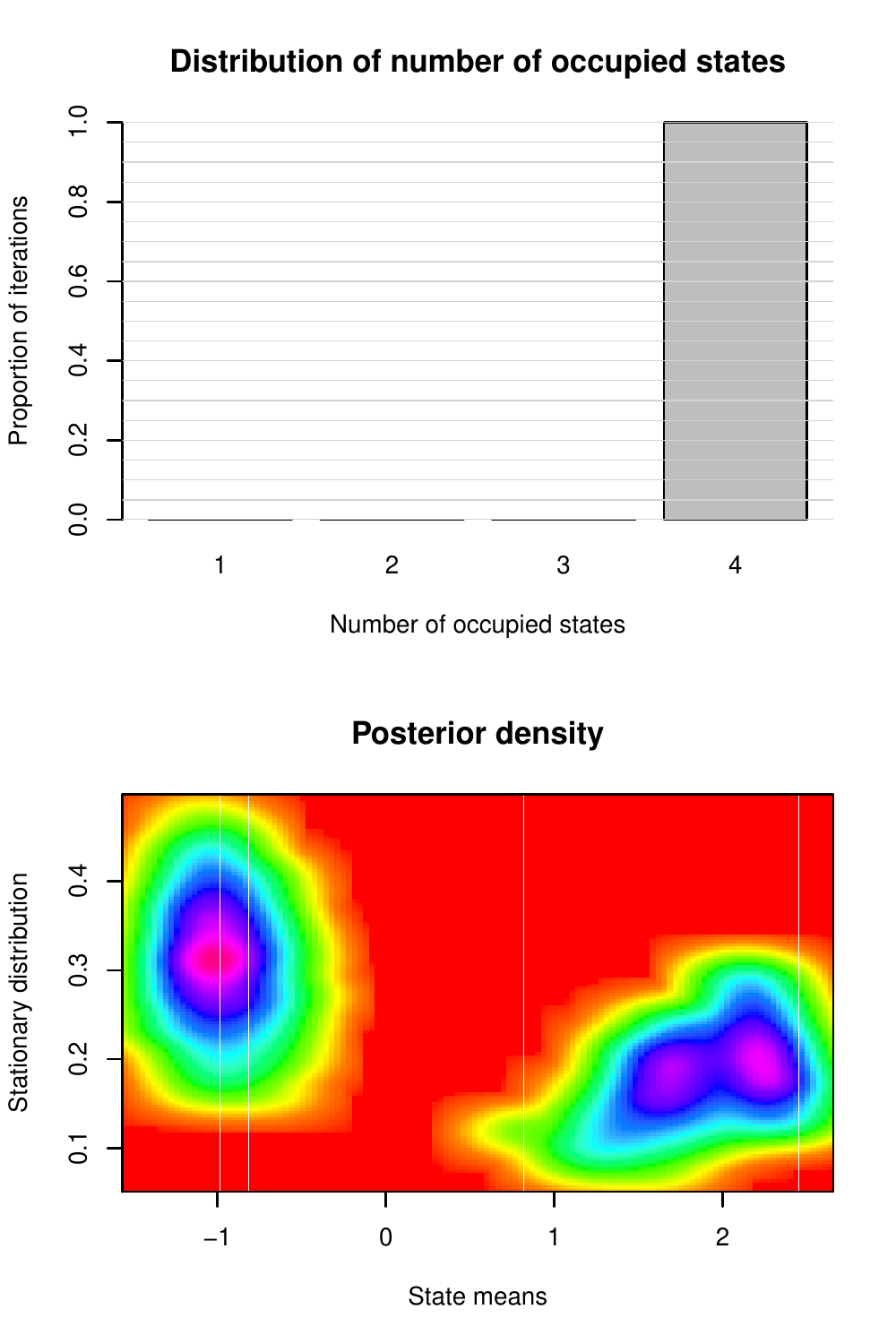}
\caption{$q_{i,.}\sim \mathcal{D}(4, 4, 4, 4)$}
\label{fig:sub2}
\end{subfigure}
  \begin{subfigure}{0.32\textwidth}
    \adjincludegraphics[trim={0 0 0 {.5\height}},clip,width=1\linewidth]{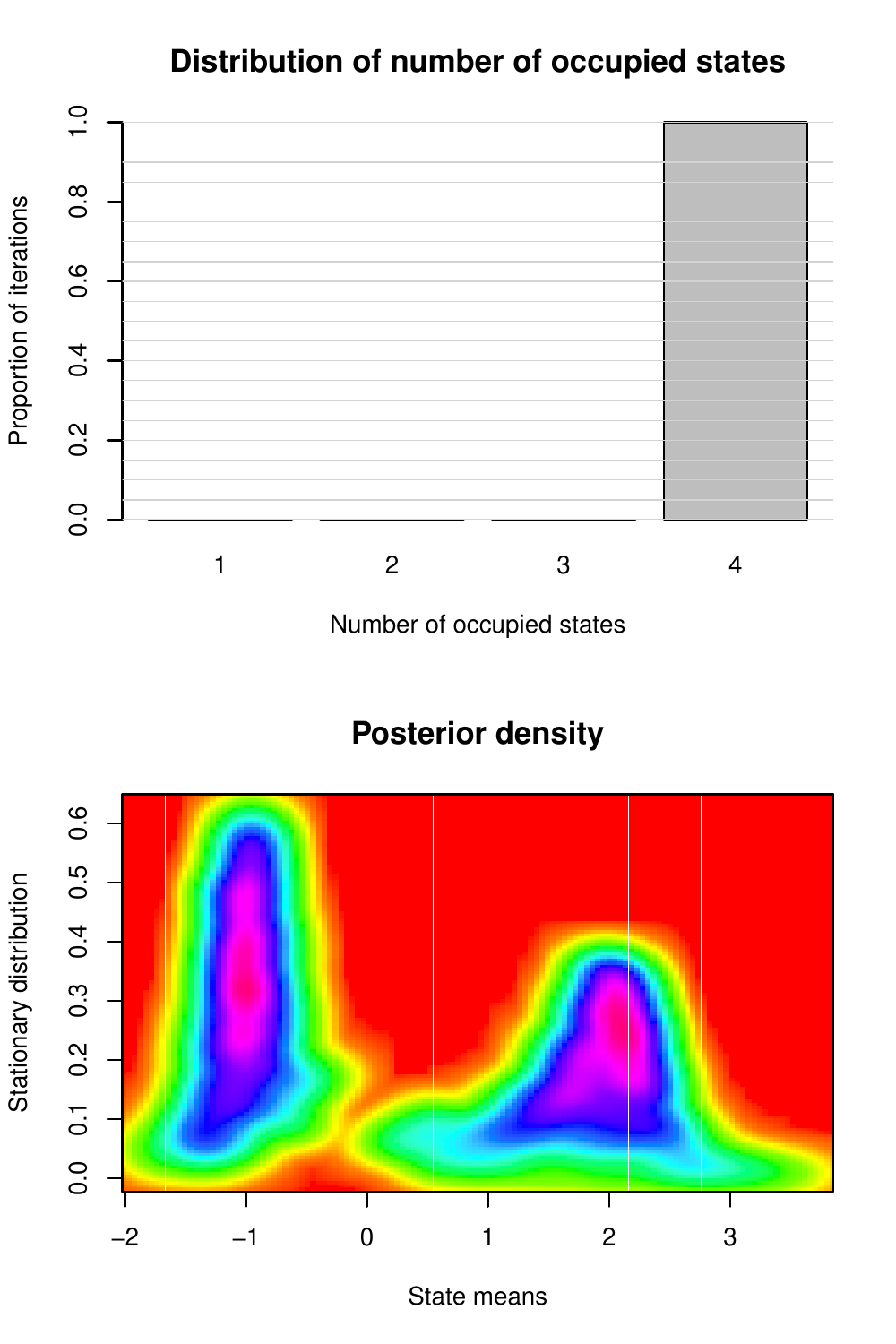}
\caption{ $q_{i,.}\sim \mathcal{D}(1, 1, 1, 1)$}
\label{fig:sub3}
\end{subfigure} \\
\begin{subfigure}{0.32\textwidth}
    \adjincludegraphics[trim={0 0 0 {.5\height}},clip,width=1\linewidth]{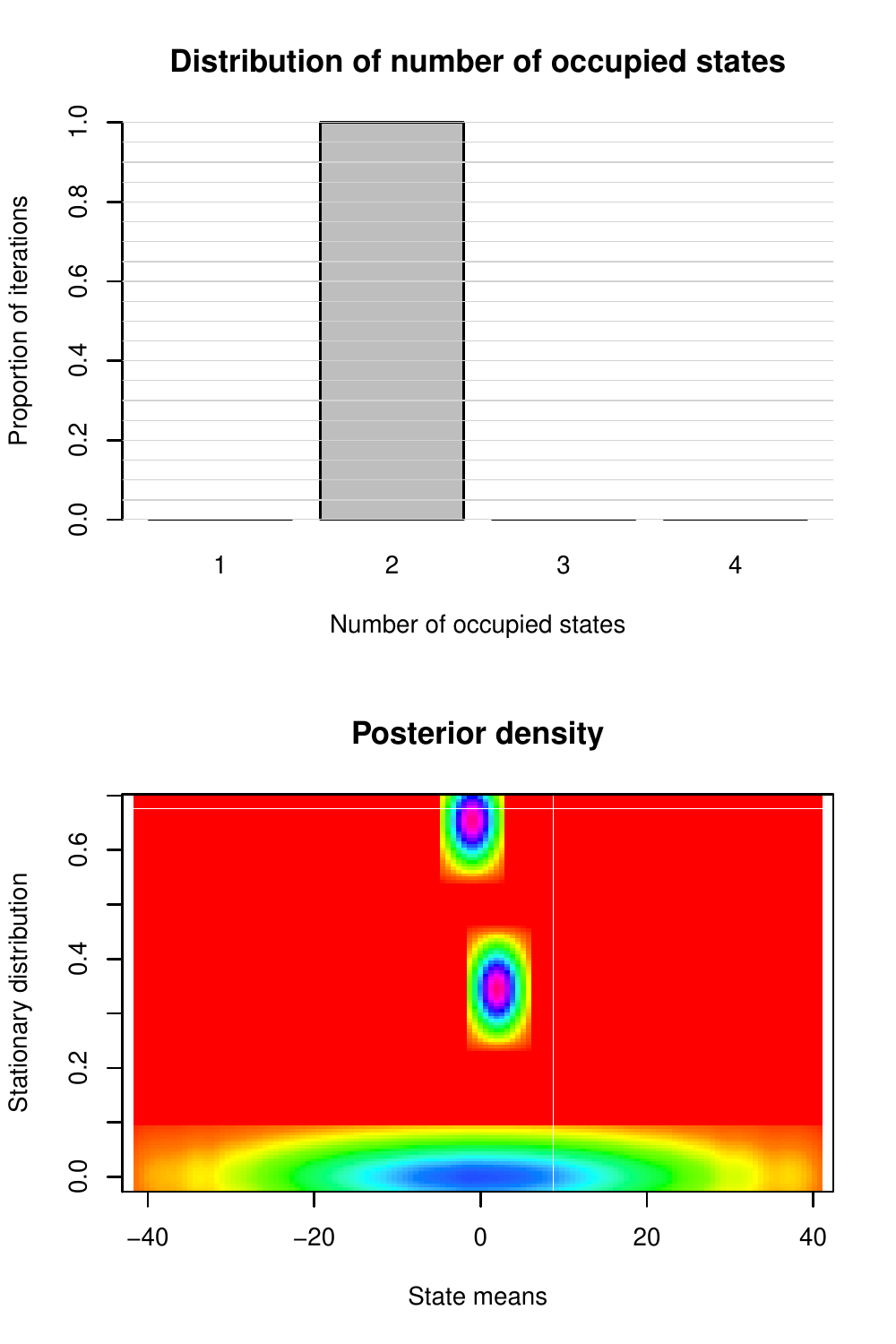}
\caption{  $q_{i,.}\sim \mathcal{D}(172,  \frac{1}{n},  \frac{1}{n},  \frac{1}{n})$}
\label{fig:sub4}
\end{subfigure}
  \begin{subfigure}{0.32\textwidth}
   \adjincludegraphics[trim={0 0 0 {.5\height}},clip,width=1\linewidth]{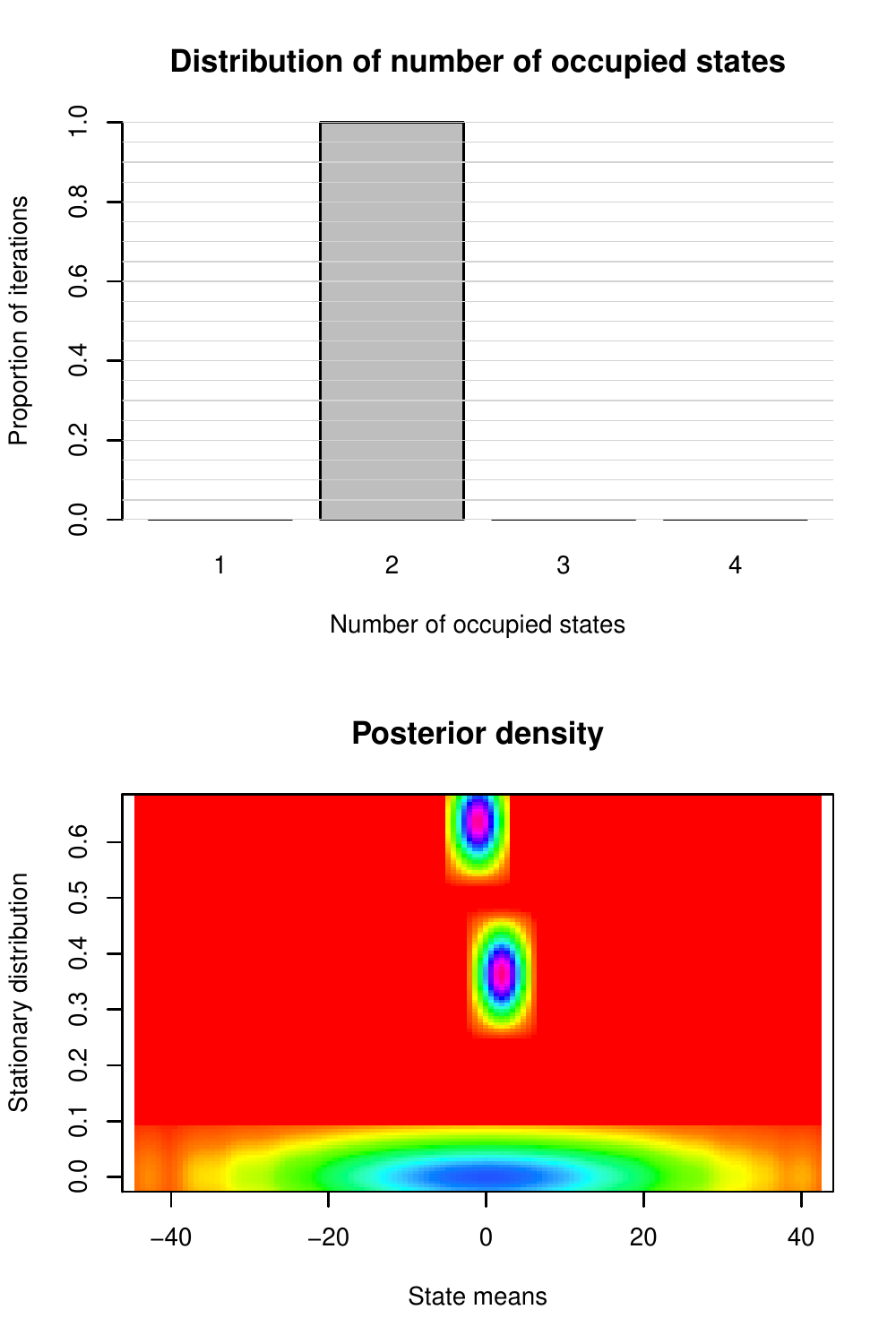}
\caption{  $q_{i,.}\sim \mathcal{D}(4, \frac{1}{n},  \frac{1}{n},  \frac{1}{n})$}
\label{fig:sub5}
\end{subfigure}
  \begin{subfigure}{0.32\textwidth}
   \adjincludegraphics[trim={0 0 0 {.5\height}},clip,width=1\linewidth]{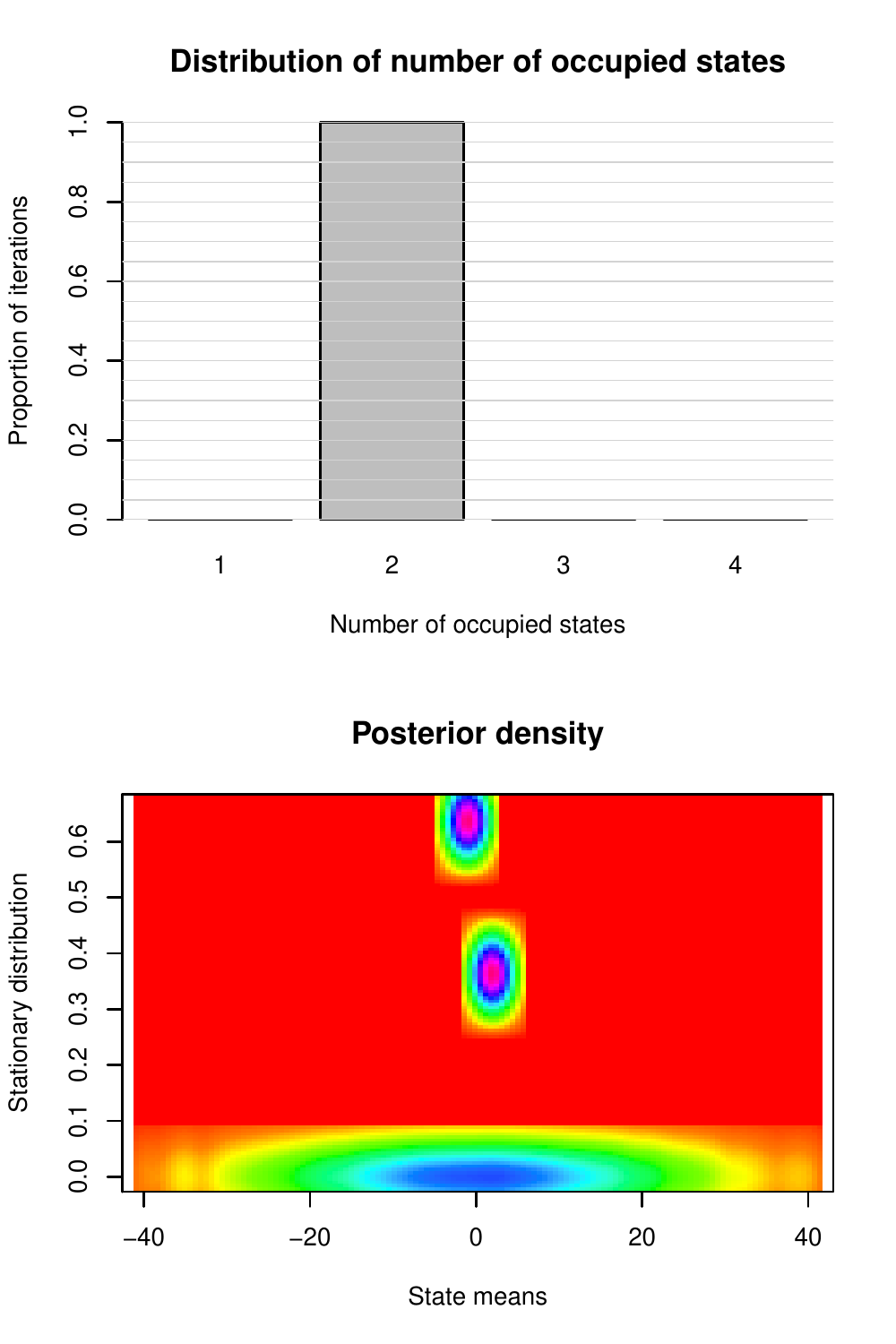}
\caption{ $q_{i,.}\sim \mathcal{D}(1,  \frac{1}{n},  \frac{1}{n},  \frac{1}{n})$}
\label{fig:sub6}
\end{subfigure}
\caption[Large sample demonstration: bivariate posterior densities]{Posterior densities produced from fitting a HMM with $K=4$ states to $n=10,000$ observations from a simulated HMM with $K^*=2$ true states. Each frame illustrates the results under different choices of the hyperparameters of the prior  $q_{i,.}\sim \mathcal{D}(\bar{\alpha}_{1},  \underline{\alpha}_{2}, \underline{ \alpha}_{3},  \underline{\alpha}_{4})$. Each column of plots refers to a different value of $\bar{\alpha}$: the theoretically validated threshold (plus 1) ($\bar{\alpha}=172$), the number of states in the model ($\bar{\alpha}=4$), and $\bar{\alpha}=1$. The rows correspond two different prior configurations: the upper row describes an overfitted posterior where $\underline{\alpha}=\bar{\alpha}$, while in the lower row of plots  $\underline{\alpha}=\frac{1}{n}=0.00001$.   Each plot illustrates the estimated bivariate density of posterior means (x-axis) and posterior stationary distribution (y-axis) over all MCMC samples for all states, representing the posterior surface sampled by the overfitted Gibbs sampler. The red region indicate areas of very low probability, increasing through orange to yellow, green, blue, then purple. }  
\label{HMM_Figure_1}
\end{figure}

\FloatBarrier
\section{Applicability of the asymptotic theory for smaller sample sizes}\label{Section3}
The problem of applying the theoretically motivated constraints on the  $\alpha_{i,j}$ prior hyperparameters  to  datasets of small sample size ($n$) is now considered.

 The influence of the hyperparameters in such cases can be unexpected and non-trivial, in part due to the intrinsic relationship that exists between the $\alpha_{i,j}$ and $n$.
 Consider that the full conditional distribution of the transition matrix, given a Dirichlet prior is placed on each row of $Q$, is $$\pi(q_{i,.}| x_{1:n} )\sim \mathcal{D}( \bar{\alpha}_{i,1}+n_{i,1}, \underline{\alpha}_{i,2}+n_{i,2}, \cdots, \underline{\alpha}_{i,K}+n_{i,K})$$ where $n_{i,j}$ is the number of transitions observed from state $i$ to state $j$. Assume as before that the states are arranged so the state with the most observations is first, and so on. 
The distribution $p(q_{i,.}| y_{1:n} )$  depends on the  $n_{i,j}$, which are intrinsically linked to the sample size. These values can be very small in practice even when an HMM is not overfitted, if transitions between two states are rare. For an overfitted HMMs with empty extra states, $n_{i,j}$ may equal 0 for $i>K_0$ and $j>K_0$.

    The hyperparameter $\alpha_{i,j}$ can be interpreted as the prior number of transitions from state $i$ to $j$, and the asymptotic constraints given in Section \ref{Section1} lead to strong statements about posterior transition probabilities when $n$ is small.
    A choice of $\underline{\alpha}=0.001$ and $K=2$ according to Theorem \ref{th:postempty}  leads to $\bar{\alpha}_{\textit{Theory}}>3.02$. $K=3$ leads to $\bar{\alpha}_{\textit{Theory}}>36.32$, and this increases rapidly as the total number of states fit to a model grows, reaching $\bar{\alpha}_{\textit{Theory}}>543.38$ when $K=5$.  For $K=10$,  as the threshold is not sharp, $\bar{\alpha}_{\textit{Theory}}>15,498.38$. These values influence the ability of an overfitted HMM to estimate the distribution of true underlying states correctly when $n$ is small, since an increasingly large sample size is required to overcome the given $\bar{\alpha}$ and estimate the true underlying $q_{i,i}$ for $i,j \leq K_0$.

    To illustrate, take as a simple example a HMM with $K^*=2$ states, and true transition probabilities $q^*_{1,1}=0.6$, $q^*_{1,2}=0.4$, $q^*_{2,1}=0.7$ and $q^*_{2,2}=0.3$.  For an arbitrary sample size $n$ and assuming the allocations are known, the transition frequencies are expected to be (approximately) $n_{i,j}=n\times q_{i,j}$ for $i,j = (1, 2)$ and 0 otherwise, following the convention of labelling by group size.
    To explore influence of the hyperparameter $\bar{\alpha}$ on the posterior transition probability $q_{1,1}$, 10,000 samples are drawn from $\pi(q_{1,}|X_t) \sim \mathcal{D}(\bar{\alpha}+n q^*_{i,1}, \underline{\alpha}_2+n q^*_{i,2}, \underline{\alpha}_3,\cdots,\underline{\alpha}_K)$, where $\underline{\alpha}_2=\cdots=\underline{\alpha}_K=0.001$). Figure \ref{HMM_Figure_2} includes box-plots of the posterior distribution of $q_{1,1}$ given $K=\{3, 5, 10\} $ and $n=\{10, 100, 1000, 10 000, 100 000, 1 000 000$ , and $10 000 000\}$. $\bar{\alpha}$ is set to three values to explore the ; $\bar{\alpha}=\textit{Theory}$, $\bar{\alpha}=K$, and $\bar{\alpha}=1$.

    Choosing $\bar{\alpha}$  according to the asymptotic bound had a very strong influence on the estimated distribution of $q_{1,1}$. When $K=3$, approximately 1,000 observations were needed for the true value of 0.5 to be within the $25^{th}$ and $75^{th}$ quantiles of the posterior distribution of $q_{i,j}$. For a model overfitted with $K=10$ states, even in this simple example, 1,000,000 observations were needed before the posterior approached the truth.   While extra states were clearly emptied by the constraints, $q_{1,1}$ was strongly skewed towards 1 and no longer corresponded to  $q_{1,1}^*$.

    Some modifications may be reasonable to obtain a less unbiased result under smaller sample sizes, as the threshold is not sharp, as mentioned. Smaller values of $\bar{\alpha}$ were observed in Section \ref{Section2} to also cause the stationary distribution of extra states to become small. Figure  \ref{HMM_Figure_2}  includes two smaller values for comparison;  $\bar{\alpha}=K=4$, and $\bar{\alpha}=1$. These resulted in a posterior distribution of $q_{1,1}$ closer to the desired value, but only results under $\bar{\alpha}=1$ included the true value for every $n$ and $K$ compared.

 \begin{figure}[htb!]
    \centering
    \includegraphics[width=\textwidth]{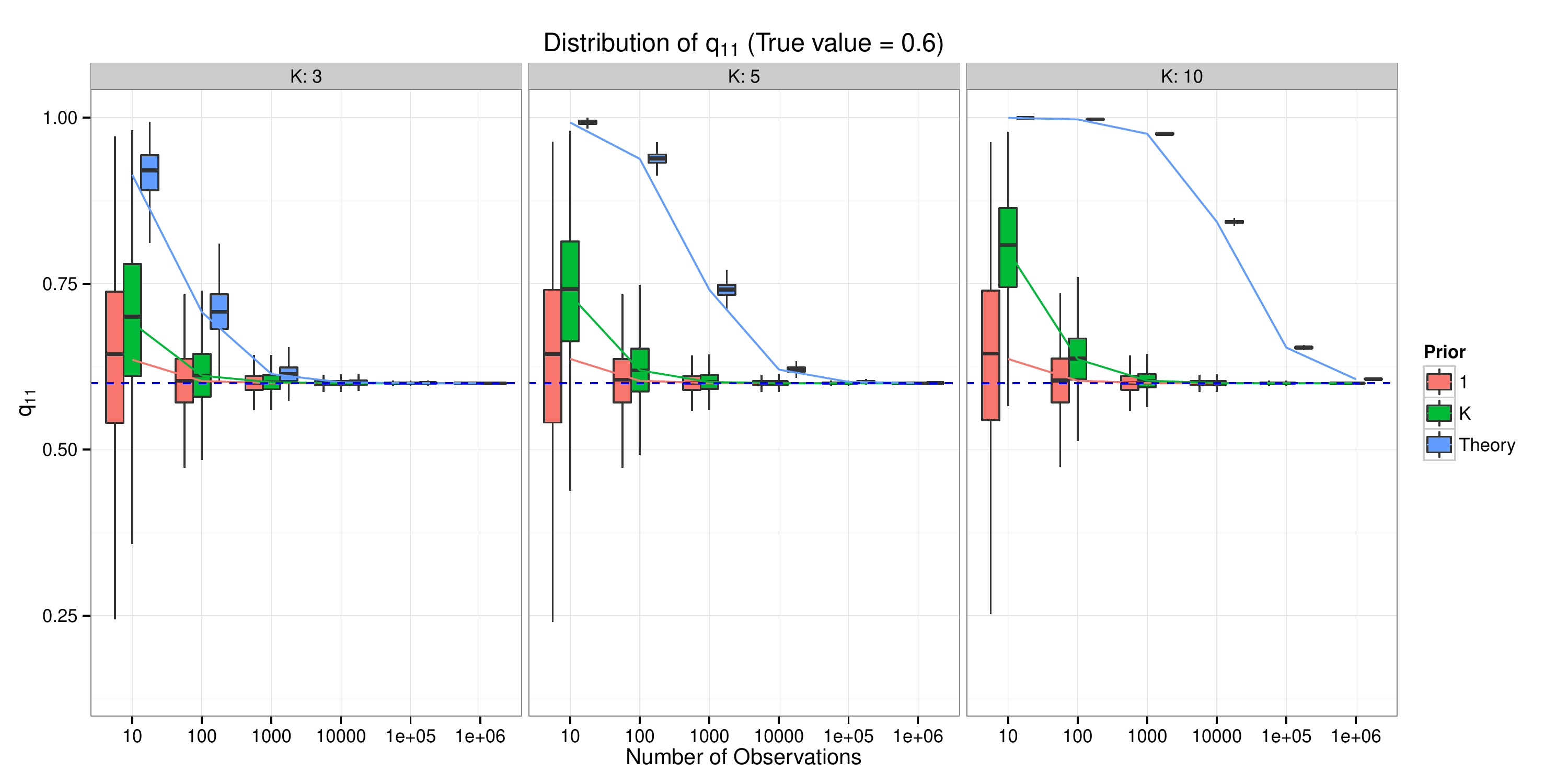}
\caption[Influence of prior on transition probabilities under different sample sizes and prior hyperpamaters.]{Posterior distribution of transition probabilities under different prior values  ($\bar{\alpha}=\text{Theory}$, $\bar{\alpha}=K$, and $\bar{\alpha}=1$) where a model with $K^*=2$ is overfitted with $K=3$, $K=5$, and $K=10$ components, under increasing sample sizes. The plots are based on 10,000 draws from $q_{i,.}\sim\mathcal{D}(\bar{\alpha}+n\times q^*_{i,1}, \underline{\alpha}_2+n\times q^*_{i,2} , \underline{\alpha}_3,\cdots,\underline{\alpha}_K )$. }
\label{HMM_Figure_2}
\end{figure}

\subsection{Considering alternative prior configurations}
The  structure of the prior is now considered in terms of the position of the large $\bar{\alpha}$ with respect to the smaller values $\underline{\alpha}$.
There are many ways a transition matrix can be written so as to accommodate the presence of extra states, because an HMM must specifically define how each extra state interacts with every other state. When seeking  to cause the extra $\mu_k$ to be small to empty out extra states, exactly how these empty states are incorporated into Q depends on the position of $\bar{\alpha}$, which determines the shape of the given prior on each row of $Q$.

\paragraph{Column Prior}
The asymptotic results developed in Section \ref{Section1} are based on a configuration where the prior is on each row of $Q$ is of the same form, which we refer to  as the \textit{Column} prior, or $\pi_{c}(Q)$, since the larger $\bar{\alpha}$ is always placed on the first position. The proof focuses on emptying the extra states with respect to the value of $\mu_Q(j)$, which is able to become small. This allows for the possibility of empty sets, which are defined through the relation $ \mu_Q = \mu_Q Q $.  $\mu_Q(j) = 0$ only if, for all $i$ such that $\mu_Q(i) \neq 0$,  the values of $q_{ij} = 0$. 

The column prior favours the configuration where $q_{ij}$ is small for $j \geq p+1$ and typically $q_{ij}$ is not small for $j\leq p$. Hence it favours the configuration which empties \emph{a priori} the last $K-p$ states. 
 This does not mean extra states are ensured to be empty in the posterior; this depends on the sample size and on the number of overfitted states included in the model, but such a framework maintains the ergodicity of the underlying Markov Chain in the presence of states with no observations.
 


\paragraph{Diagonal Prior}
Another type of asymmetrical prior formulation has been used in a non-parametric setting where HMMs model DNA segmentation \citep{Nur2009,Boys2000,Boys2004}. This is obtained by placing a large hyperparameter $\bar{\alpha}$ on the diagonal of the prior on Q, instead of the first position, the motivation being that while the number of states is unknown, the Markov chain is expected to be more likely to remain in the same state than to transition to another state \emph{a priori}. 

Compared to the column prior, a diagonal prior (referred to as  $\pi_d(Q)$) favours transition matrices $Q$ with small values off-diagonal, describing a model where states are unlikely to transition to the extra states. This could lead to non ergodic transition matrices, where emptying of extra states is possibly obtained. 
Obvious extensions apply when there are blocks of states.  While this raises concerns about the ability of the Markov chain created to mix effectively, if there is more than one stationary distribution one alternative is to choose the one with the largest mass; if there are multiple modes, then any one can be chosen at random.

\paragraph{Mixture Prior}
A third prior is considered to attempt to leverage the benefits of both approaches while overcoming their drawbacks by drawing on both structures. This is a mixture of both priors, referred to as a mixture prior $\pi_m(Q)$, where:
\[\pi_m(Q)\sim 0.5\pi_{c}(Q)+ 0.5\pi_{d}(Q)\]

This leads to three potential priors for overfitted HMMs, $\pi_c({Q})$, $\pi_d({Q})$, and $\pi_m({Q})$, which are compared in Section \ref{Section4}.

\subsection{Mixing difficulties and Prior Parallel Tempering}
The use of asymmetrical hyperparameters can also create additional computational difficulties, such as a severe lack of mixing over the MCMC. Multi-modal posteriors such as those created by HMMs with emptied out extra states have a high propensity for lack-of-mixing, as large areas of low probability can separate supported modes. Successful MCMC estimation of HMMs is dependent on obtaining well mixed samples which are able to reach all areas of the posterior space to identify those of higher probability, and convergence is not assured for a badly mixed MCMC sampler \citep{Celeux2012}.

 This lack of mixing is a familiar problem with MCMC estimation of the related finite mixture model when this is overfitted, and a method exists which attempts to take advantage of the unique topography of these models to overcome this: Prior Parallel Tempering, developed for overfitted finite mixture models by \citet{vanHavre2015} .
Parallel tempering involves the creation of a number of parallel Gibbs samplers which operate on ever more smoothed versions of the posterior distribution, sharing information occasionally, allowing the MCMC to overcome mixing problems with some added computational cost.

We adapt the Prior Parallel Tempering (PPT) method from \citet{vanHavre2015} to the case of overfitted HMMs. PPT obtains its smoothed posteriors by allowing the smaller hyperparameter $\underline{\alpha}$ to increase slowly, which causes the empty states to gradually merge with occupied states, softening the extreme posterior surface. For HMMs, the value of $\underline{\alpha}^{j}$ is raised incrementally until it matches $\bar{\alpha}$ in the parallel samplers.

 The tempering is incorporated by allowing $J$ blocked Gibbs samplers to run with slightly different hyperparameters in the prior on $Q$. The MCMC samplers are allowed to share information by swapping posterior samples when these are close. A Metropolis-Hastings acceptance step is needed to judge this, which is defined in Equation \ref{eq:PPThmms}. For HMMs, the acceptance ratio for this move will depend on a function of the probability of the initial distribution (assumed to be the stationary distribution) as well as $Q$. Given $p(\mu_Q)p(Q)=P(X^0)\prod_{i=1}^{K} p(q_{i,.})$, the acceptance ratio for a tempering move can be derived in the same way as in \citet{vanHavre2015}, retaining all terms involved with the prior on the transition matrix $Q$.

\begin{equation} \label{eq:PPThmms}
A_{PPT_{ij}}=\dfrac{\pi(Q^{i}|\alpha^{j}) \pi(Q^{j}|\alpha^{i})   }
  { \pi(Q^{i}|\alpha^{i})    \pi(Q^{j}|\alpha^{j})  } 
\end{equation}

The number of chains and the distance between the $\underline{\alpha}^{(j)}$ is chosen arbitrarily so that they are not too far apart; in practice the acceptance ratio is tracked to ensure plentiful mixing is occurring between all PPT chains.

\section{Verification: small sample simulation.}\label{Section4}
  This section considers a set of simulated HMMs with small sample sizes, which are overfitted under three proposed priors configurations and various hyperparameter values, under a Gibbs sampling scheme augmented by PPT. The interest is twofold; first, to provide some approximate consistency results for priors for which no theoretical results exist, and secondly to explore in more detail the influence of the choice of prior configuration and hyperparameter values.  Specifically, the matter of whether they induce appropriate posterior emptying of extra states, how well the resultant models represent the simulated data, and whether the true underlying parameters can be retrieved are queried.

  The estimation methodology is now described, followed by an overview of the simulations considered, and the strategy for analysis. This is comprised of both a replicate simulation study and a detailed but single-realization illustrative study.

\subsection{Methods}\label{sec:HMM_methods2}
  The prior on the rows of the transition matrix $Q$ can take three forms which have been described in detail. They are referred to as the column prior $\pi_c(q_{i,.})$, the diagonal prior $\pi_d(q_{i,.})$, and the mixture prior $\pi_m(q_{i,.})$, such that:
    \begin{itemize}
   \item  $\pi_c(q_{i,.})\sim \mathcal{D}( \bar{\alpha}, \underline{\alpha} )$
   \item  $\pi_d(q_{i,.}) \sim \mathcal{D}(\underline{\alpha}, \bar{\alpha},   \underline{\alpha} )$
   \item $\pi_m(q_{i,.}) \sim 0.5\pi_{c}(q_{i,.})+ 0.5\pi_{d}(q_{i,.})$
   \end{itemize}

  Prior Parallel Tempering is adapted to HMMs and incorporated into the Gibbs sampler described in Section \ref{Section2} in a similar manner as was done in \citet{vanHavre2015}.

\subsubsection{Gibbs sampler with Prior Parallel Tempering for HMMs}

  \begin{enumerate}
      \item Initialise $J$ tempered chains:
      \begin{enumerate}
      \item Choose prior values for the tempering: $\bar{\alpha}$, and $\underline{\alpha}^{j}$ for $j=1,\cdots, J$,
          \begin{enumerate}
          \item $\bar{\alpha}$ is the largest value of the hyperparameter.
          \item  $\underline{\alpha}^{j}$ for $j=1,\cdots, J$ describe the smaller hyperparameter in each chain:
              \begin{itemize}
              \item $\underline{\alpha}^{1}=\bar{\alpha}$.
              \item $ \underline{\alpha}^{j\prime} = \bar{\alpha}\left( \frac{\bar{\alpha}}{\underline{\alpha}_j}  \right)^{-\frac{(j-1)}{J}}   \quad \text{for } j=2,\cdots, J-1$  
              \item $\underline{\alpha}^{J}$ is the smallest, and the prior of the \textit{target} chain.
              \end{itemize}
          \end{enumerate}
      \item Choose a set of starting values for the allocations for each chain,   $x_{1:n}^{ j,0 }$
      \end{enumerate}
      \item Step $m$, for each iteration $m=1,\cdots, M$:
            \begin{enumerate}
              \item \textit{Gibbs Sampling Step}. For each chain $j=1, \cdots, J$,
                      \begin{enumerate}
                      \item Generate transition matrix $Q^\prime$ given previous states from $\pi(Q^{j,m}| x_{1:n}^{j,{(m-1)}})$,
                      \begin{itemize}
                      \item Accept $Q^\prime$ as new value for $ Q^{j,m}$ with probability $min \left(1,\frac{\mu_{Q^{(m-1)}} \left(X_1^{j,m}\right)}{\mu_{Q^\prime}\left(X_1^{j,m} \right) } \right) $

                      \end{itemize}
                      \item Generate $\gamma_k^{j^m}$ from $\pi(\gamma_k^{j,m}|y_{1:n}, x_{1:n}^{j, m-1})$ for each $k=1,\cdots, K$
                      \item Generate $x_{1:n}^{j,m}$ from $ p_n(x_{1:n}|Q^{j,m},\gamma^{j,m}, y_{1:n})$ using the forward backward algorithm.
                      \end{enumerate}

             \item \textit{Tempering Step}. Randomly choose $z_0=1$ or $2$. Then from $z = z_0$ to $z= J$, with probability             
  $$ A(z)=\dfrac{\pi(Q^{z,m}|\alpha^{z+1})  \pi(Q^{z+1,m}|\alpha^{z})   }
  { \pi(Q^{z,m}|\alpha^{z}) \pi(Q^{z+1,m}|\alpha^{z+1})  } $$
  accept tempering move and exchange:
             \begin{itemize}
             \item  $\gamma^{z,m}$ with  $\gamma^{z+1,m}$,
             \item  $x_{1:n}^{z,m}$ with $x_{1:n}^{z+1,m}$,
             \item  $Q^{z,m}$ with  $Q^{z+1,m}$.
             \end{itemize}

      \end{enumerate}
    \end{enumerate}

\subsubsection{Simulations}

A number of simulations are now described. These are designed to test the impact of the prior on overfitted HMMs of various configurations in practice, under small sample sizes. Plots of samples of size $n=100$ from each simulation considered are included in Figure \ref{HMM_Figure_3}. Simulation 1 (Sim 1) has $K^*=3$ states,  one of which is relatively well separated from two which overlap. Simulation 2 (Sim 2) described another HMM with $K^*=3$ states but under different state means and transition probabilities. Simulation 3 (Sim 3) contains $K^*=5$ states, with relatively well separated, equally spaced means, and mainly large values on the diagonal of $Q^*$. This is chosen to  describe a HMM with `sticky' behaviour, where the Markov chain is likely to remain in the same state. The true parameters of each simulation are detailed below.  All emission distribution  have known variances equal to 1.
  \begin{description}
  \item[Sim 1]  $\gamma_{S1}^*=(1,3,6)$,   $\mu_{S1}^*=( 0.33, 0.38, 0.29)$, and
                                       $Q^*_{S1}=\left[\begin{array}{ccc}
                                       0.2&0.3 &0.5  \\
                                       0.5 & 0.25 & 0.25\\
                                       0.25& 0.65& 0.1
                                        \end{array} \right] $
  \item[Sim 2]  $\gamma_{S2}^*=(-5, 5, 9)$, $\mu_{S2}^*=( 0.56, 0.18, 0.26)$, and $Q^*_{S2}=\left[\begin{array}{ccccc}
              0.8 & 0.1 & 0.1\\
              0.2 & 0.4 & 0.4\\
              0.3 &  0.2 &  0.5 \\
              \end{array} \right] $ ,
  \item[Sim 3]  $\gamma_{S3}^*=(-10, -5, 0, 5, 10)$,  $\mu_{S3}^*=(0.11, 0.24, 0.20, 0.22, 0.22)$, \\and  $Q^*_{S3}=\left[\begin{array}{ccccc}
                  0.2 & 0.3 &0.1 &0.2 &0.2 \\
                  0.1 & 0.6 &0.1 &0.1 &0.1 \\
                  0.1 & 0.1 &0.6 &0.1 &0.1 \\
                  0.1 & 0.1 &0.1 &0.6 &0.1 \\
                  0.1& 0.1& 0.1& 0.1  &  0.6\\
                  \end{array} \right] $.
  \end{description}

  \begin{figure}[htb!]
  \centering
  \includegraphics[width=1\textwidth]{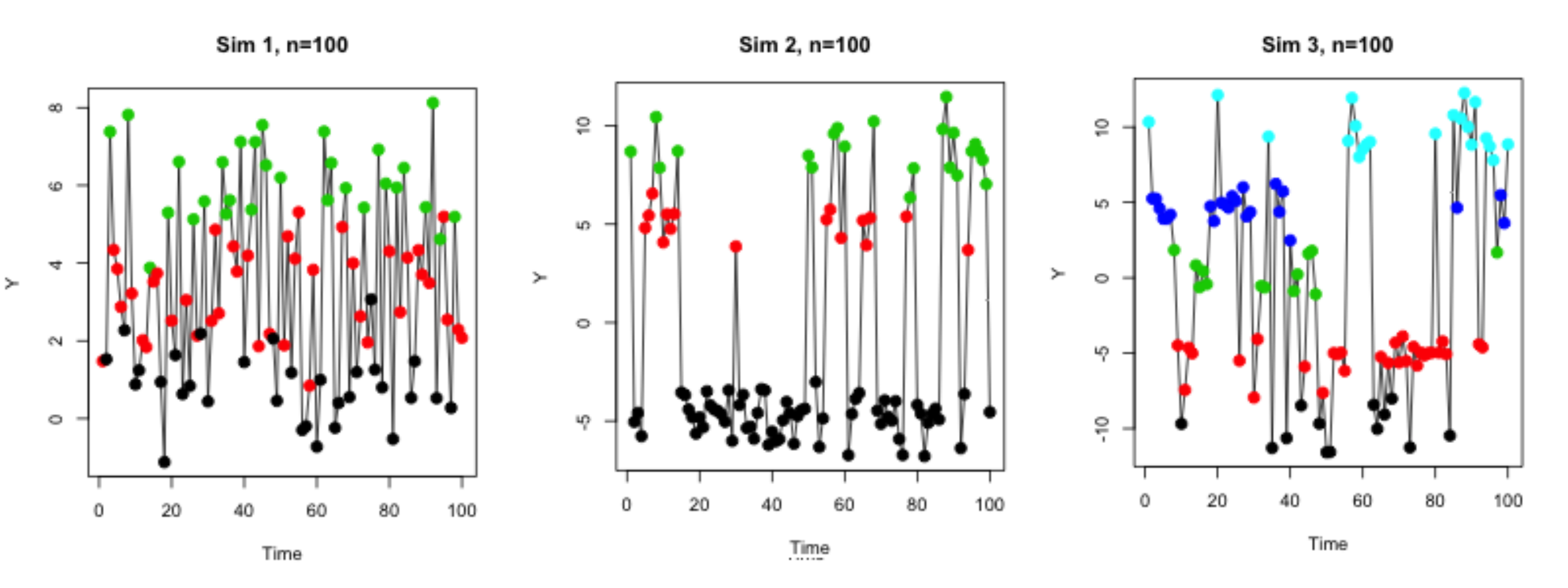}
  \caption[HMM simulations: Observations and true states for one realization of n=100]{Realizations of $n=100$ observations from the three simulations considered in this study. The points are coloured according to their true state of origin.}
  \label{HMM_Figure_3}
  \end{figure}

\subsubsection{Strategy}
  The analysis proceeds in two parts, as follows.
  \begin{description}
    \item[Replicate simulations]
      Sim 1 is used as the basis of a replicate simulation study (25 replicates) to explore the proposed priors in a small sample case, for $n=100$ and $n=500$. For each replicate, $K=10$ states are fit under the three prior configurations: the column prior $\pi_c(q_{i,.})$, the diagonal prior $\pi_d(q_{i,.})$, and the mixture prior $\pi_m(q_{i,.})$. For each prior type, all combinations of the following hyperparameters are used: $\bar{\alpha}=(n , K, \text{ and } 1)$, and $\underline{\alpha}=(\frac{1}{n}, \text{ and } \frac{1}{10n})$.   This is modelled with a Gibbs sampler with 30 PPT chains, for 20,000 iterations, discarding the first 10,000 as a burn-in. For each replicate, the empirical mode of the number of occupied states is computed. 

    \item[Illustrative simulations]  To illustrate the use of overfitting over a wider variety of simulations, a number of HMMs are modelled under each prior type ($\pi_c(q_{i,.})$,  $\pi_d(q_{i,.})$, and $\pi_m(q_{i,.})$). Hyper-parameters are set to $\bar{\alpha}= 1$ and $\underline{\alpha}=\frac{1}{n}$.
    To perform the analysis, samples of $n=100$ observations are drawn from Sim 2 and Sim 3. A sample of $n=500$ is also drawn from Sim 3.
    Each dataset thus created is modelled with $K=10$ states using a Gibbs sampler with 30 PPT chains, for 20,000 iterations ( discarding the first 10,000 as a burn-in).

    For each model, the distribution of the number of occupied states is explored, and the iterations associated with the most frequent value extracted for further processing. Label switching is resolved and parameter summaries, various quality of fit statistics and plots are created from the output.
      This includes the mean absolute error (MAE) and mean squared error (MSE), and statistics based on the posterior predictive distribution, estimated by resampling the posterior samples of the parameters in the MCMC in order to predict $10,000$ datasets of size $n$. Mean absolute errors (MAPE) and mean squared prediction errors (MSPE) are reported as an average over the replicates. Predictive concordance is also computed, which can be interpreted as the average proportion of $y_i$'s that are not outliers given the model (based on the suggestion that any $y_i$ that is in either $2.5\%$ tail area of $y_i^{rep}$ should be considered an outlier) \citep{Gelfand1996}. An ideal fit should have a predictive concordance of around $95\%$.
  \end{description}

%
%
%
%
%
%

\FloatBarrier
\subsection{Results}
\subsubsection{Replicate Simulation Study}

For each replicate, the empirical mode of the distribution of occupied states, $\hat{K_A}$, was computed, and the distribution of $\hat{K_A}$ is detailed in Table \ref{tab:HMM_Table_2}.
When $n=100$, the distribution of $\hat{K}_A$ is tightly concentrated around 2 occupied states for all priors, and only a few replicates under any prior configuration reported the correct number.
The best performing was the column prior $\pi_c$ with $\bar{\alpha}=1$ and $\underline{\alpha}=1/n$, where 33\% of replicates were able to identify the correct number of states.
The results were concentrated more strongly on small values of $K_A$ for the larger choices of $\bar{\alpha}$ and the smaller $\underline{\alpha}$.

For $n=500$, the three states were more clearly defined in the simulations, and while the column prior with $\bar{\alpha}=1$ reported 3 occupied states with a probability of 1,  the results revealed some perturbing inconsistencies under the other priors.

First, the larger value of $\bar{\alpha}=n$ caused the number of occupied states to be consistently underestimated. For  $\bar{\alpha}=n$, $\pi_c$ reported the true number of occupied states (3) in 11\% of replicates only, and 2 otherwise; both other priors resulted in 2 occupied states in  100\% of replicates.
Reducing the hyperparameter value to $\bar{\alpha}=K$ (where $K=10$) did not appear to be  sufficient to prevent this for most priors: only $\pi_c$ with $\underline{\alpha}=1/n$ identified  3 occupied  states in 89\% of replicates.
Second, the type of prior was associated with the consistent estimation of the correct number of occupied clusters. The posterior behaviour was stable when $\bar{\alpha}=1$ for the column prior $\pi_c$, where all replicates reliably reported 3 occupied states for both values of $\underline{\alpha}$ essayed. This was also true of the mixture prior.
The results from the diagonal prior were unstable, and 50\% of replicates resulted in a majority of only 2 occupied states for both  $\underline{\alpha}=1/n$ and $\underline{\alpha}=1/10n$ ( for the later, 16\% reported 4 occupied states). This may be due either to estimation difficulties of intrinsic mixing issues with the Markov Chain suggested by the prior.

  \begin{table}[htb!p]
  \centering
  \caption[Results of replicate simulation study.]{\emph{Replicate study of Sim 1: Proportion of replicates which contained $\hat{K_A}=j$ occupied states for $j=2,3,4$. Unobserved values equal to exactly zero are shown as a dash. All probabilities $>0.5$ which refer to the true number of states are indicated in bold.}}
\resizebox{\textwidth}{!}{%
  \begin{tabular}{lll|rr|rrr|rr}
    \multicolumn{  3}{r}{\textbf{Prior type:} }      &    \multicolumn{  2}{c}{$\pmb{\pi_c}$}                       &                \multicolumn{  3}{c}{ $\pmb{\pi_d}$}&\multicolumn{  2}{c}{ $\pmb{\pi_m}$}                 \\ \hline
  $n$   & $\bar{\alpha}$ & $\underline{\alpha}$ &  $P(\hat{K_A}=2)$  & $P(\hat{K_A}=3)$  & $P(\hat{K_A}=2)$ & $P(\hat{K_A}=3)$ & $P(\hat{K_A}=4)$ & $P(\hat{K_A}=2)$ & $P(\hat{K_A}=3)$ \\ \hline
  100 & n    & 1/$n$        &  0.920       & 0.078       & 1.000         & -        & -         & 1.000        & -         \\
  100 &      & 1/10$n$      &  0.930    & 0.065      & 1.000         & -       & -        & 1.000        & -         \\ \hline
  100 & 10   & 1/$n$        &  0.842       & 0.158        & 1.000         & -       & -        & 0.980        & 0.022       \\
  100 &      & 1/10$n$      &  0.970       & 0.031      & 1.000         & -         & -        & 1.000        & -        \\ \hline
  100 & 1    & 1/$n$        &  0.630       & 0.350       & 0.940  &  0.061       & -        & 0.790 & 0.210        \\
  100 &      & 1/10$n$       &  0.910      & 0.091       & 1.000         & -       & -        & 1.000        & -         \\ \hline

  500 & n    & 1/$n$         & 0.871       & 0.129        & 1.000         & -       & -        & 1.000        & -         \\
  500 &      & 1/10$n$      &  0.920      & 0.083        & 1.000         & -       & -          & 0.690 & 0.310         \\ \hline
  500 & 10   & 1/$n$        &  0.400       & \textbf{0.600}        & 0.930        & 0.070    & -        & 0.670        & 0.330        \\
  500 &      & 1/10$n$      &  0.714       & 0.286        & 1.000         & -      & -        & 0.667        & 0.333        \\ \hline
  500 & 1    & 1/$n$        &  -           &\textbf{1.000}        & 0.500         & 0.500   & -         & -            &\textbf{ 1.000}            \\
  500 &      & 1/10$n$      &  -           & \textbf{1.000}       & 0.440         & 0.440   & 0.110     & 0.110  & \textbf{0.890}    \\
  \hline
  \end{tabular}}
  \label{tab:HMM_Table_2}
  \end{table}

\FloatBarrier
\subsubsection{Illustrative Simulation}
The three small-sample simulations described in Section \ref{sec:HMM_methods2} were modelled with $K=10$ states according to the strategy described and $\bar{\alpha}=1$, $\underline{\alpha}=1/n$.

\begin{figure}[htb!p]
    \centering
    \begin{subfigure}{0.32\textwidth}
    \adjincludegraphics[trim={0 0 0 {.5\height}},clip,width=1\linewidth]{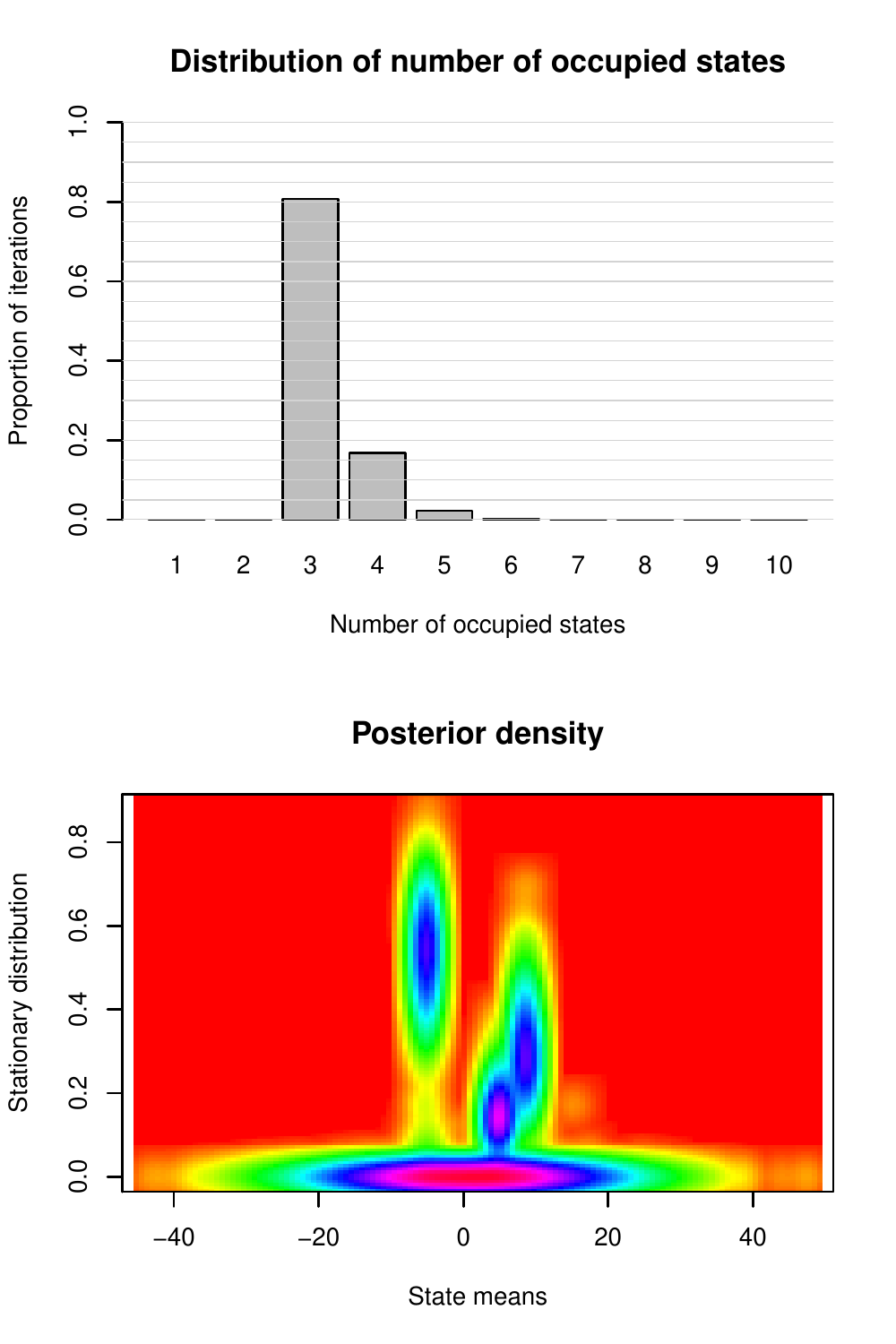}
    \caption{$\pmb{\pi_c}$,  Sim 2}
    \end{subfigure}
      \begin{subfigure}{0.32\textwidth}
    \adjincludegraphics[trim={0 0 0 {.5\height}},clip,width=1\linewidth]{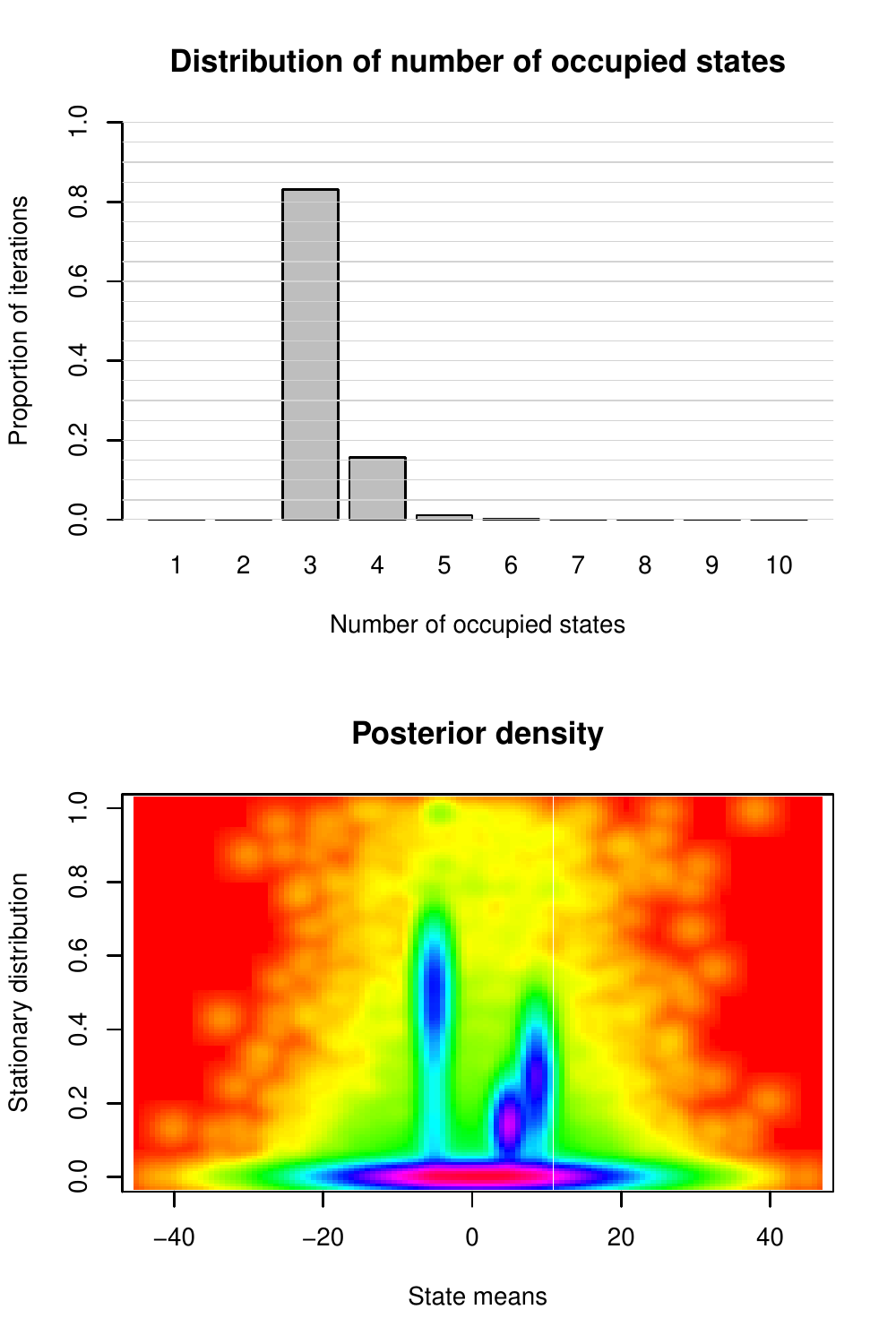}
    \caption{$\pmb{\pi_d}$,  Sim 2}
    \end{subfigure}
    \begin{subfigure}{0.32\textwidth}
    \adjincludegraphics[trim={0 0 0 {.5\height}},clip,width=1\linewidth]{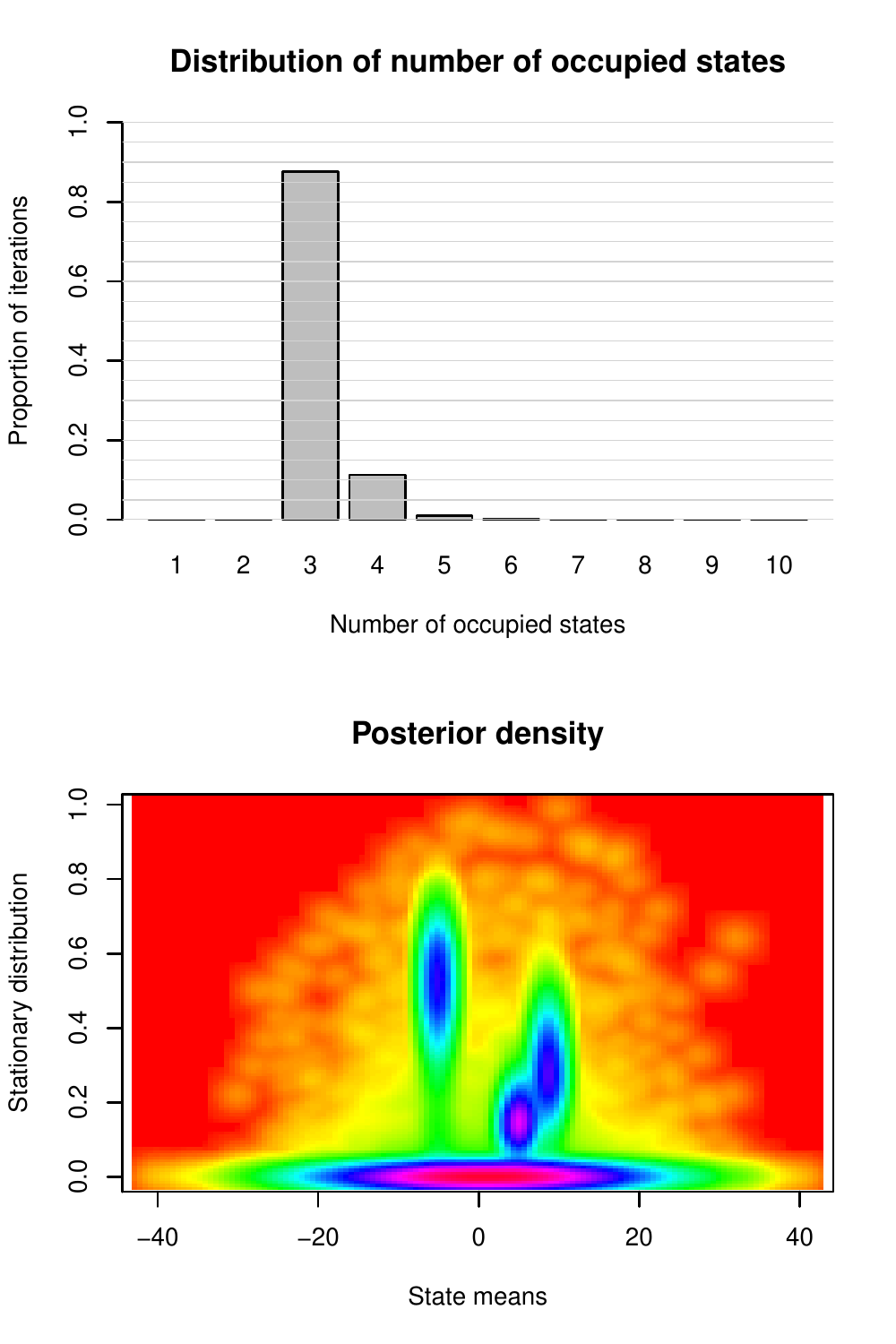}
    \caption{$\pmb{\pi_m}$, Sim 2}
    \end{subfigure}\\
    \begin{subfigure}{0.32\textwidth}
    \adjincludegraphics[trim={0 0 0 {.5\height}},clip,width=1\linewidth]{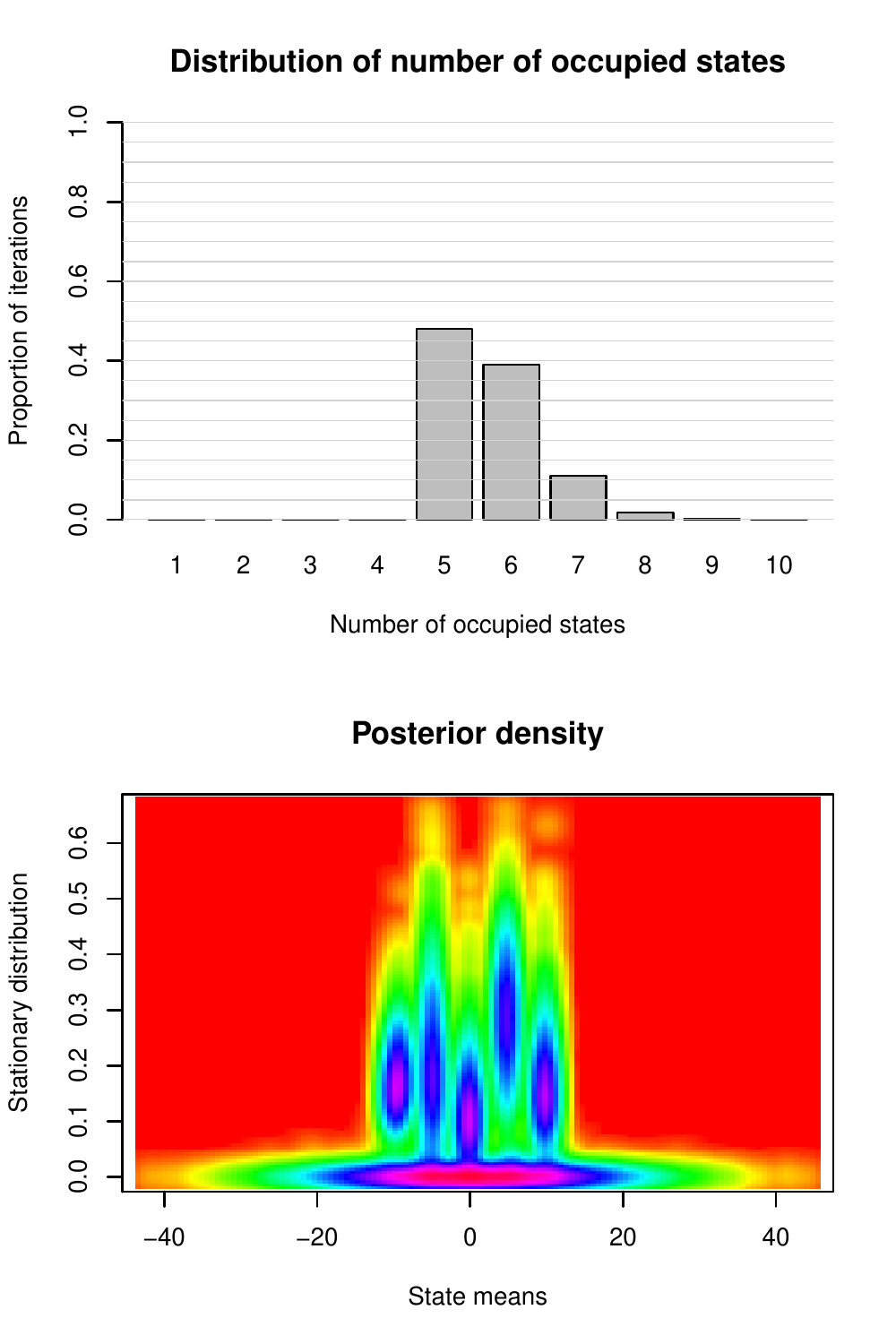}
    \caption{ $\pmb{\pi_c}$, Sim 3}
    \end{subfigure}
      \begin{subfigure}{0.32\textwidth}
    \adjincludegraphics[trim={0 0 0 {.5\height}},clip,width=1\linewidth]{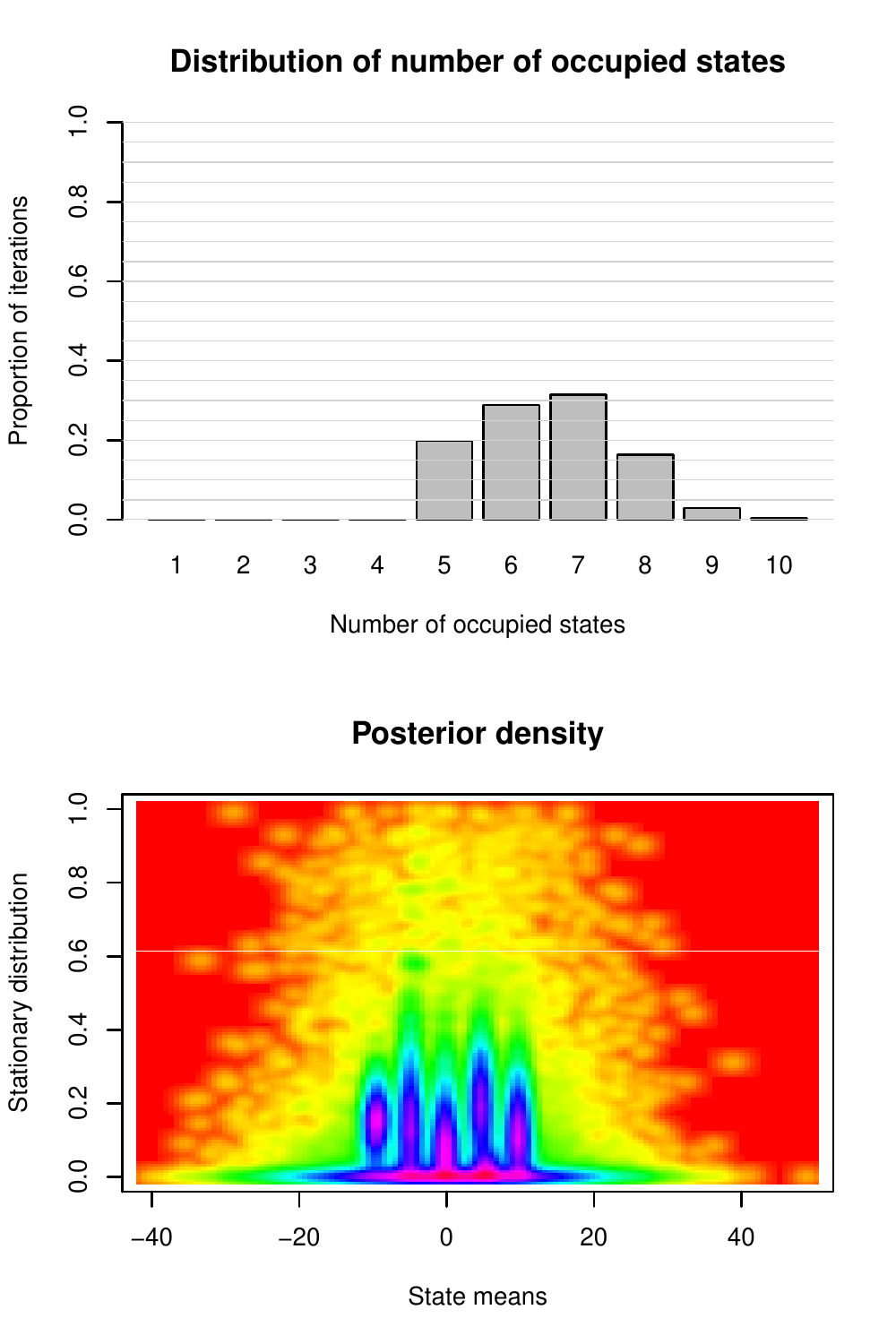}
    \caption{$\pmb{\pi_d}$, Sim 3}
    \end{subfigure}
    \begin{subfigure}{0.32\textwidth}
    \adjincludegraphics[trim={0 0 0 {.5\height}},clip,width=1\linewidth]{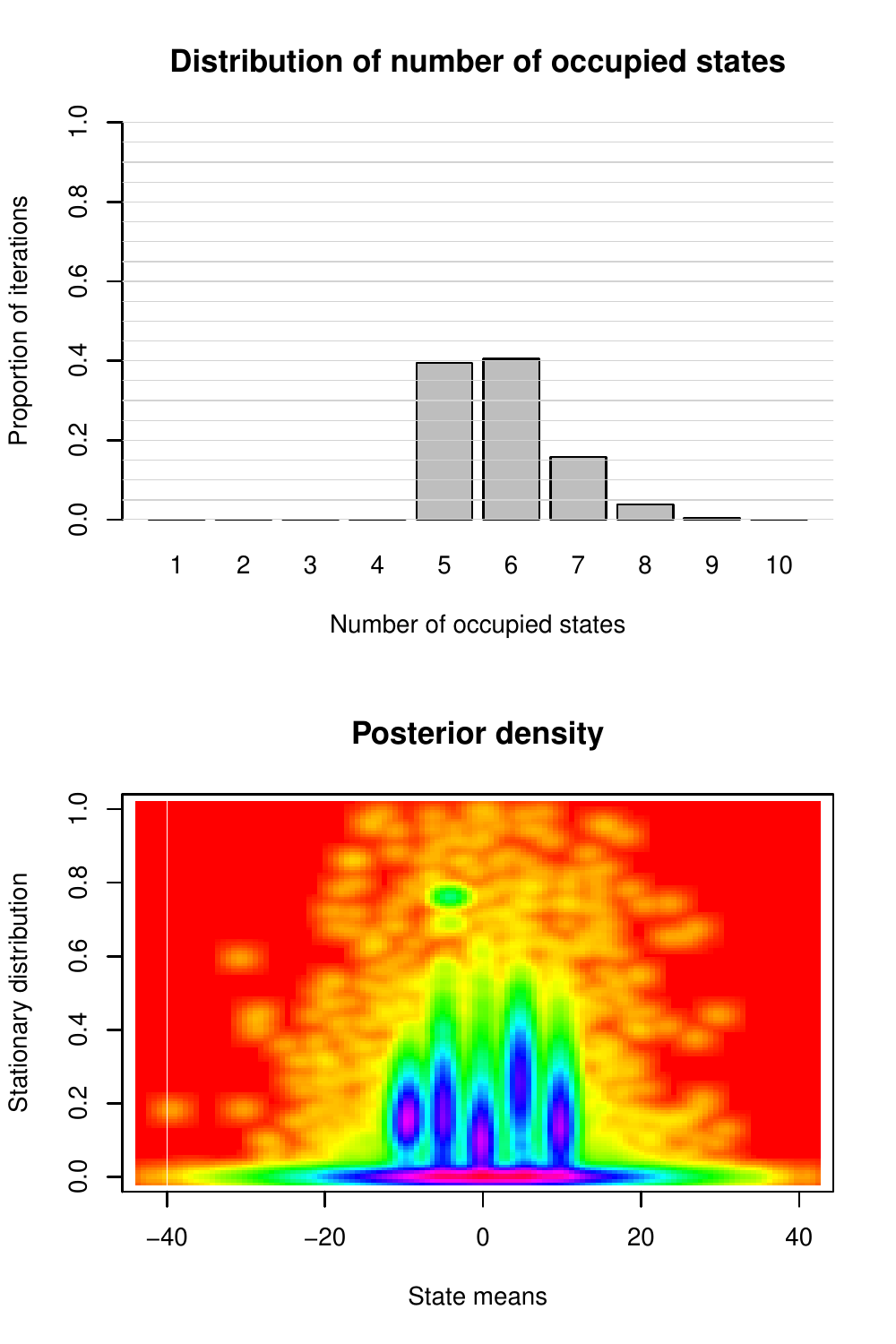}
    \caption{$\pmb{\pi_m}$, Sim 3}
    \end{subfigure}
  \caption[Small sample demonstration: bivariate posterior densities]{\emph{\textbf{Posterior densities produced from fitting a HMM with K=10 states to n=100 observations from Sim 2 and Sim 3.} Each row includes a plot for each choice of prior, where $\bar{\alpha}=1$ and $\underline{\alpha}=1/n$. Each plot illustrates the posterior surface sampled by the tempered Gibbs sampler.} }
  \label{fig:HMM_Figure_5}
    \end{figure}

\begin{table}[htb!]
\centering
\caption[Table of Goodness-of-fit statistics, illustrative simulations]{\emph{ \textbf{Quality-of-fit statistics for Sim 2 (n=100) and Sim 3 (n=100, and 500), fit with K=10 using a column prior $\pmb{\pi_c}$ with $\pmb{\bar{\alpha}=1}$ and $\pmb{\underline{\alpha}=1/n}$. }Includes, from left to right, the simulation (Sim), the number of occupied states ($K_A$), and the proportion of iterations reporting this value, $P(K_A)$. MAE and MSE are the Mean Squared and Absolute Errors respectively. $\%$ is the percentage of observations correctly reclassified into their state of origin. \textit{Conc.} is the concordance, and the average Mean Absolute Prediction Errors (MAPE) and average Mean Squared Prediction Errors (MSPE) are also included.}}
\label{tab:HMM_Table_3}
\begin{tabular}{rllrrrrrrrr}
  \hline
& Sim  & $K_A$ & $P(K_A)$ & \% & MAE    & MSE      & Conc. & MAPE & MSPE & \\
  \hline
 &  2 (n=100) & 3 & 0.81 & 0.96 & 77.78 & 91.78  & 0.94 & 81.45 & 104.1 &\\
 &  2 (n=100) & 4 & 0.17 & 0.68 & 185.62 & 595.18 & 0.94 & 94.72 & 186.12   &\\ \hline
 &  3 (n=100) &  5 & 0.48 & 0.97 & 86.89 & 128.45  & 0.96 & 82.79 & 107.51  &\\
 &  3 (n=100) &  6 & 0.39 & 0.55 & 510.45 & 5920.65  & 0.96 & 105.04 & 262.72  &\\
 &  3 (n=100) &  7 & 0.11 & 0.17 & 623.39 & 5006.68  & 0.96 & 134.06 & 407.58   &\\ \hline
 &  3 (n=500) & 5 & 0.93 & 0.98 & 404.41 & 523.01  & 0.95 & 401.18 & 505.65  &\\
 &  3 (n=500) & 6 & 0.07 & 0.12 & 2284.87 & 11827.61 & 0.95 & 1143.70 & 5170.52   &\\
   \hline
\end{tabular}
\end{table}

The inconsistencies which were noted in the replicate simulation study under the diagonal prior can be further clarified by observing the posterior parameter space under the three prior structures. Using the results of illustrative simulation study, Figure \ref{fig:HMM_Figure_5} includes 2-dimensional density plots of the stationary distributions $\mu_k$ and the means $\lambda_k$, for Sim 2 (n=100) and Sim 3 (n=500) under the column, mixture, and diagonal priors ($\pi_c$,$ \pi_d$, and $\pi_m$). While the high-probability modes were modelled similarly by the three priors, $ \pi_d$ and to some extend $\pi_m$ caused the posterior to also contain a large number of spurious samples not associated to other modes, visible as a cloud of yellow points. This illustrates one of the dangers of using a prior which is not theoretically justified; while in some instances, the result will be correct, there is no guarantee that some other result will not be reported instead. In this example, for Sim 2, the mode of $K_A$ and the posterior parameter estimates of $\hat{\mu}$ and $\hat{\lambda}$ were very similar between the three priors, but the same was not true for Sim 3 (with n=100). Here there were  a greater number of occupied components than expected, caused by merging between some states, and $\hat{K_A}=7$ for $\pi_d$ and  $\hat{K_A}=6$ for $\pi_m$.
Figure \ref{fig:HMM_Figure_6} includes the distribution of the number of occupied states for each  posterior space represented in Figure  \ref{fig:HMM_Figure_5}.

\begin{table}[htb!]
\centering
\caption[Table of Parameter estimates, Sim 2, n=100]{\emph{\textbf{Parameter estimates: Sim 2 (n=100), fit with K=10 using a column prior $\pmb{\pi_c}$ with $\pmb{\bar{\alpha}=1}$ and $\pmb{\underline{\alpha}=1/n}$}: Parameter estimates for $K_A=3$, ordered by increasing state means and including 95\% Credible Intervals. } }
\begin{tabular}{rlrrrrr}
\hline
& Parameter   & k=1                     & k=2                   & k=3                  &\\ \hline
& $\mu_k$     & 0.54 (0.35, 0.71)     &  0.16  (0.08,  0.26) &  0.31 (0.16,  0.47)   &\\
& $\gamma_k$ &  -5.09 (-5.36,  -4.82)   & 4.96 (4.36,  5.57)   &  8.64  (8.24, 9.05)& \\
\hline
\end{tabular}
\label{tab:Sim2pars}
\end{table}

For Sim 2 with $n=100$, an HMM was fitted with 10 states with a column prior.  Three states were occupied for the majority of iterations (81\%), with 17\% containing 4 states, and negligible weight on larger values. Comparing the two potential candidate models proposed individually ( with $K_A=3$ and $K_A=4$), the MSE and MAE strongly point towards the model containing the correct number of states (see  Table \ref{tab:HMM_Table_3} ) .
Focusing on the model with $K_A=3$, the estimated parameters were close to the true values, and the means and allocation probabilities estimated with little uncertainty (included in  Table \ref{tab:Sim2pars}). These are also illustrated in Figure \ref{fig:ResSim2}.  Only 3 observations were incorrectly clustered compared with their states of origin for this simulation.

Sim 3 contained five true states, which were difficult to distinguish when $n=100$. Here, the distribution of $K_A$ was flatter than that seen for Sim 2, and spread over a wider range of values; while 48\% of iterations reported 5 occupied states, 39\% of the MCMC sampling period was spent exploring a configuration with $K_A=6$ states, with a further 11\% on 7 states (see Table \ref{tab:HMM_Table_3}).
Comparing the two most commonly reported models (Table \ref{tab:HMM_Table_2})

\begin{table}[h]
\centering
\caption[Table of Parameter estimates, Sim 3]{\emph{\textbf{Parameter estimates: Sim 3 (n=100 and 500), fit with K=10 using a column prior $\pmb{\pi_c}$ with $\pmb{\bar{\alpha}=1}$ and $\pmb{\underline{\alpha}=1/n}$}: Parameter estimates, ordered by increasing state means and including 95\% Credible Intervals. }}
\resizebox{1.1\textwidth}{!}{%
\begin{tabular}{rlrrrrrrr}
\hline
$n$ & Parameter      & $k=1  $                 & $=2 $                  & $=3    $              & $=4$               & $=5 $       \\
\hline
100 & $\mu_k$      & 0.18 (0.09, 0.29)     & 0.22 (0.08, 0.42)    & 0.13 (0.04,  0.27)  &  0.30 (0.14,  0.47) &  0.17  (0.07, 0.32)\\
500 & $\mu_k$      & 0.11  (0.08, 0.15)  &   0.19   (0.13, 0.26)  &  0.19   (0.14, 0.26)&  0.33   (0.25, 0.41)  &  0.18   (0.13, 0.25) \\ \hline
100 & $\gamma_k$  &-9.53(-10.27,-8.84)  & -4.95(-5.42 , -4.48 )& -0.17 (-0.74,  0.39) & 4.73   (4.29 ,  5.16) & 9.74 (9.28,  10.19)  \\
500 & $\gamma_k$  & -10.19 (-10.46, -9.91) &  -5.03  (-5.25 , -4.81) &  0.11 (-0.10,  0.32) &   4.88   (4.72,5.04) & 10.10 (9.88,  10.30)&    \\
\hline
\end{tabular}}
\label{tab:Sim3Pars}
\end{table}

The results of fitting 10 states to samples of size 500 from Sim 3 identified 5 occupied states with high probability ($P(K_A=5)=0.93$), tightly concentrated around the true value with only 7\% of iterations reporting 6 occupied states (see Table \ref{tab:HMM_Table_3}). Table \ref{tab:Sim3Pars} contains the mean and 95\% CI for the parameters estimated for the 5 state models; these were close to the truth, and decrease in variance for $n=500$.

Graphical summaries of the best model for each simulation are included in Figure \ref{fig:BestModelsHMM}.  While relatively simple with well separated clusters, the simulations demonstrated  the effectiveness of the column prior for emptying extra states and allowing the retrieval of the original parameter estimates, which are not skewed by the prior hyperparameters.

\begin{figure}[htb!p]
    \centering
   \begin{subfigure}{0.9\textwidth}
   \makebox[\textwidth][c]{
    \adjincludegraphics[trim={0 {.5\height} 0 0},clip,height=3cm ]{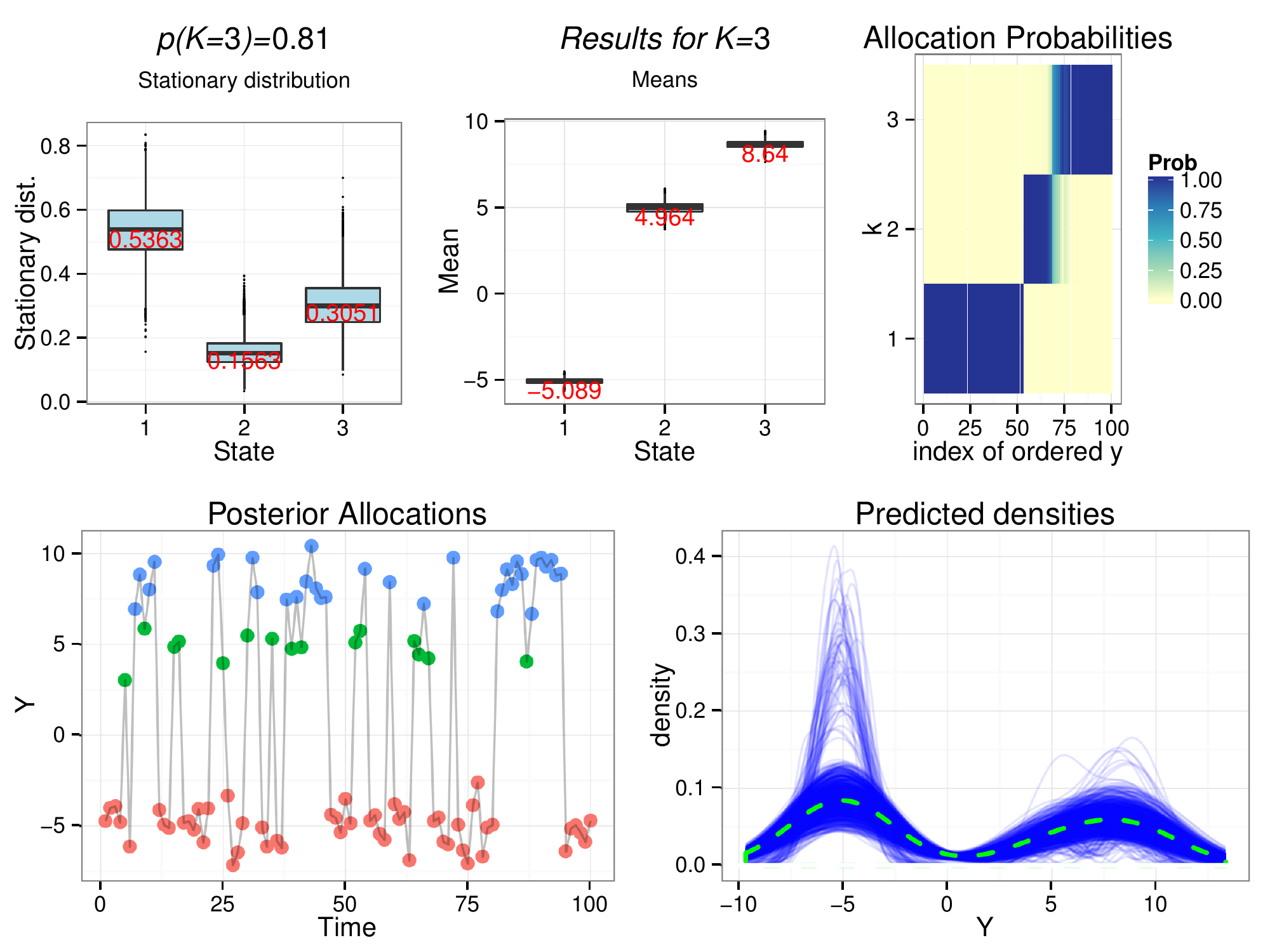}
    \adjincludegraphics[trim={0 0 0 {.5\height} },clip,height= 3cm]{Fig5_1}}
    \caption{Sim 2 (n=100).}
    \label{fig:ResSim2}
    \end{subfigure}
    \begin{subfigure}{0.9\textwidth}
    \makebox[\textwidth][c]{
    \adjincludegraphics[trim={0  {.5\height} 0 0},clip,height=3cm]{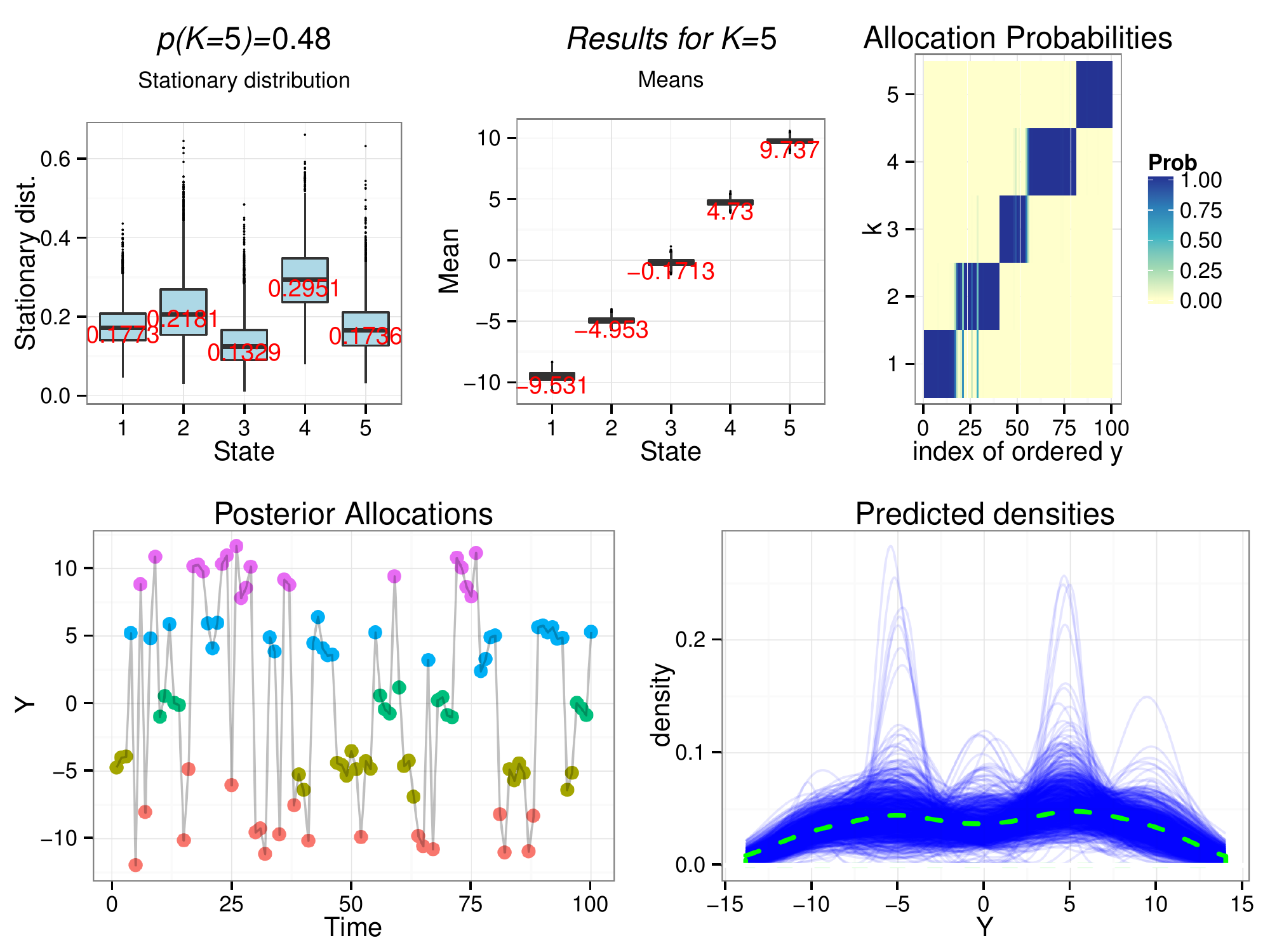}
    \adjincludegraphics[trim={0  0 0 {.5\height}},clip,height=3cm]{Fig5_2}}
    \caption{Sim 3 (n=100)}
    \label{fig:ResSim3n100}
    \end{subfigure}
    \begin{subfigure}{0.9\textwidth}
    \makebox[\textwidth][c]{
    \adjincludegraphics[trim={0 {.5\height} 0 0  },clip,height=3cm]{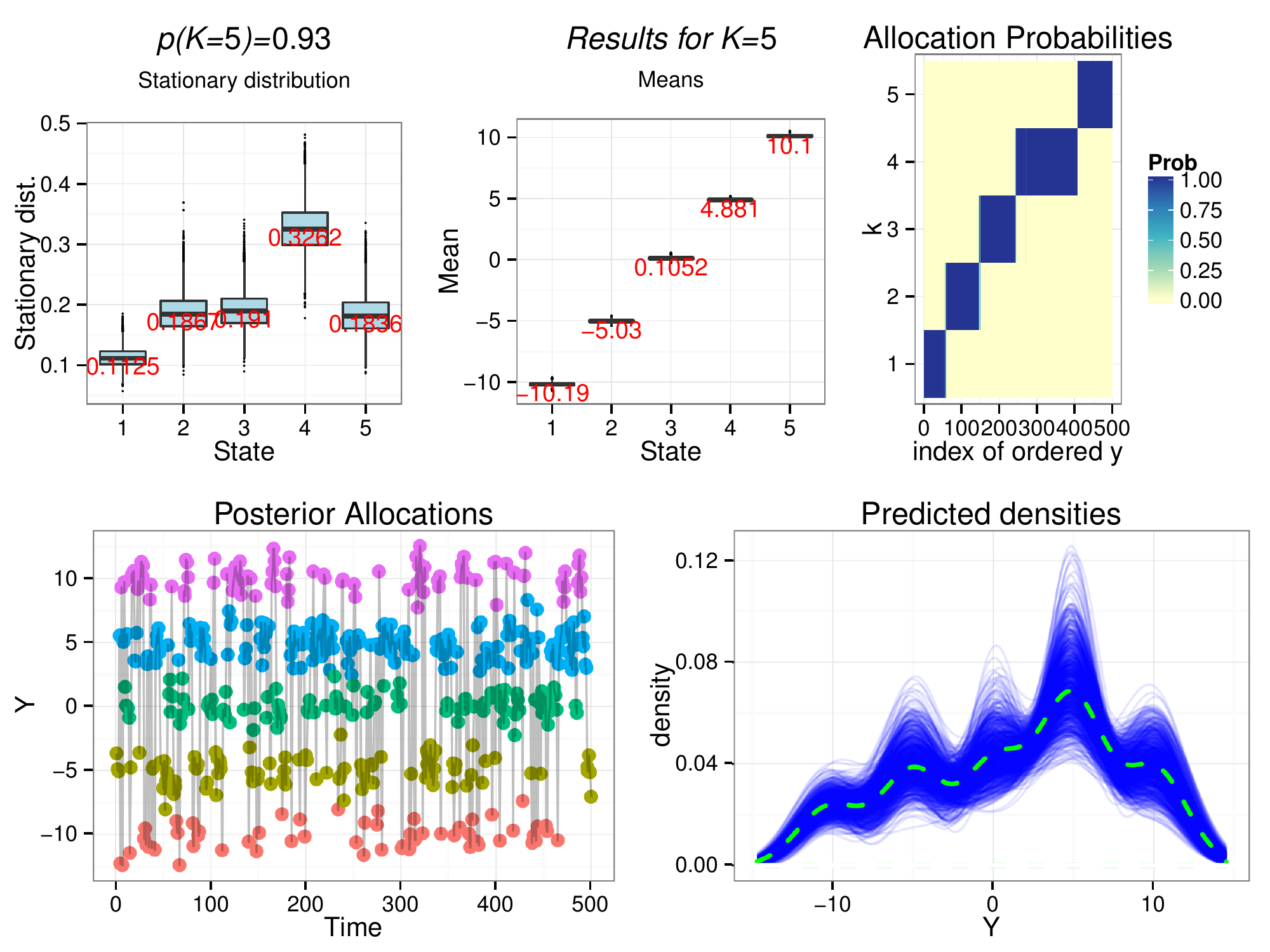}
    \adjincludegraphics[trim={0 0 0 {.5\height} },clip,height=3cm]{Fig5_3}}
    \caption{Sim 3 (n=500)}
    \label{fig:ResSim3n500}
    \end{subfigure}
     \caption[Posterior parameter estimates and allocation probabilities for illustrative simulations]{
     \emph{\textbf{Graphical summary of most frequented model for Sim 2 (n=100) and Sim 3, (n=100 and 500), obtained from fitting $K=10$ states with prior $\pi_c \sim \mathcal{D}(\bar{\alpha}_1, \underline{\alpha}_2, \dots, \underline{\alpha}_K )$, with $\bar{\alpha}=1$ and $\underline{\alpha}=1/n$.  } From left to right (in figures a-c): distribution of $\mu_k$ and $\gamma_k$ for non-empty states, posterior allocation probabilities for each state, the original data  coloured according to estimated state, and a density plot of the data (green dashed line) over  that of 10,000 datasets predicted from the posterior.
     The proportion of MCMC iterations reporting the  $K_A$ pertaining to the model in question is indicated in the upper left hand corner of the figure. }}
     \label{fig:BestModelsHMM}
 \end{figure}

\section{Discussion}
Overfitting HMMs in such a way as to empty out the stationary distribution of extra states was proven theoretically and shown to be possible in practice on large and small sample sizes. While the number of occupied states was not proven to be a consistent estimator of the true number, careful choice of hyperparameters were observed to encourage extra states to be rarely allocated observations in practice, by allowing the stationary distribution of extra states to be smaller than the weight of an observation. We suggest one should choose $\underline{\alpha}=\frac{\alpha_0}{n}$ for some $\alpha_0$; then, possibly $\hat{K}_A$ becomes consistent and simulations included point toward this, but this is only conjecture for the moment. In practice, we recommend checking the posterior samples and distribution of the number of components for evidence $\underline{\alpha}$ was low enough, flat or broad distributions over $K_A$ are strong indication this is not the case.

The value of $\bar{\alpha}$ dictated by the asymptotic constraints caused concerns in practice due to the relationship between this hyperparameter and the number of components. In order to be conservative on a vague boundary, Theorem \ref{th:postempty} resulted in very large
lower bounds for $\bar{\alpha}$, which are quite informative with respect to the posterior transition probabilities of occupied states when $n$ is not very large with respect to $K$. Retaining the general shape of the prior while softening the bounds was found to be equally effective at emptying out the $\mu_k$ relating to extra states when $n$ is not large, while allowing the posterior distributions of the occupied states to reflect the true state distributions.

The diagonal prior was  inconsistent and caused the MCMC and the Markov chain estimated to behave oddly, occasionally becoming trapped in configurations where one or two extra states merged with the true states. While it could result in the correct posterior, it appeared to be equally likely to produce another version with more (or fewer) occupied states. The mixture prior was able to overcome these problems to a certain extent, but did not perform better than the theoretically justified column prior. If a diagonal prior must be used for a specific application, such as certain genotyping studies \citep{Boys2000,Nur2009}, our results indicated that care should to taken, and models should be run several times to ensure the results are consistent and free of merged states. Compared to the diagonal prior, a mixture prior should lead to the same target posterior more consistently as the involvement of the column prior allows Markov chain to evolve out of bad, merged configurations.

While the theoretical proof is quite general, as is the proposed Gibbs sampler with Prior Parallel Tempering, the sections of this paper which referred to simulations were restricted to Gaussian HMMs where the state specific variances $\sigma^2_i$ were assumed to be known and equal to 1. The results are expected to apply equally well to Gaussian HMMs with unknown variances and known means. More work will be required to allows Gaussian HMMs with unknown means and unrestricted state specific variances to be reliably overfitted, however. In preliminary experiments, we found the added complexity results in a posterior which invariably chooses too few states, with large variances, unless they are extremely well separated. It will likely be possible to impose some stronger restrictions on the prior of the state variances  or state means to prevent such behaviour.

In conclusion, asymptotically and for very large sample sizes, clear constraints supported by theory were shown to exist on the prior on transition probabilities, which cause the posterior of the stationary distribution of extra states to be arbitrarily small while allowing the Markov chain thus created to remain ergodic. These constraints translate to strong statements on certain prior transition probabilities which become unreasonable when smaller, more realistic sample sizes are considered. The solution proposed here is to select values of $\bar{\alpha}$ and $\underline{\alpha}$ which soften the prior distribution while retaining its shape. The column prior, such that
$$\pi_c(q_{i,.})\sim \mathcal{D}( \bar{\alpha}, \underline{\alpha}, \underline{\alpha}, \cdots) $$
combined with Prior Parallel Tempering was observed to be able to consistently induce posterior emptying of extra states and allow the retrieval of posterior parameter estimates of the occupied states, which, once label switching was resolved and given a sufficient $n$, approximated the true parameter values closely.
\FloatBarrier

\FloatBarrier

\bibliographystyle{ba}

\begin{thebibliography}{32}
\newcommand{\enquote}[1]{``#1''}
\expandafter\ifx\csname natexlab\endcsname\relax\def\natexlab#1{#1}\fi
\expandafter\ifx\csname url\endcsname\relax
  \def\url#1{{\tt #1}}\fi
\expandafter\ifx\csname urlprefix\endcsname\relax\def\urlprefix{URL }\fi
\ifx\endbibitem\undefined \let\endbibitem\relax\fi

\bibitem[{Beal et~al.(2002)Beal, Ghahramani, and Rasmussen}]{Beal2002}
Beal, M.~J., Ghahramani, Z., and Rasmussen, C.~E. (2002).
\newblock \enquote{{The Infinite Hidden Markov Model}.}
\newblock {\em Neural Inf. Process. Syst. Conf.\/}, 14: 577--585.
\endbibitem

\bibitem[{Boys and Henderson(2004)}]{Boys2004}
Boys, R.~J. and Henderson, D.~a. (2004).
\newblock \enquote{{A Bayesian approach to DNA sequence segmentation.}}
\newblock {\em Biometrics\/}, 60(3): 573--578.
\endbibitem

\bibitem[{Boys et~al.(2000)Boys, Henderson, and Wilkinson}]{Boys2000}
Boys, R.~J., Henderson, D.~A., and Wilkinson, D.~J. (2000).
\newblock \enquote{Detecting homogeneous segments in DNA sequences by using
  hidden Markov models.}
\newblock {\em Journal of the Royal Statistical Society: Series C (Applied
  Statistics)\/}, 49(2): 269--285.
\endbibitem

\bibitem[{Celeux et~al.(2000)Celeux, Hurn, and Robert}]{Celeux2012}
Celeux, G., Hurn, M., and Robert, C.~P. (2000).
\newblock \enquote{{Computational and Inferential Difficulties with Mixture
  Posterior Distributions}.}
\newblock {\em Journal of the American Statistical Association\/}, 95(451):
  957--970.
\endbibitem

\bibitem[{Chambaz et~al.(2009)Chambaz, Garivier, and Gassiat}]{Chambaz2009}
Chambaz, a., Garivier, a., and Gassiat, E. (2009).
\newblock \enquote{{A minimum description length approach to hidden Markov
  models with Poisson and Gaussian emissions. Application to order
  identification}.}
\newblock {\em Journal of Statistical Planning and Inference\/}, 139(3):
  962--977.
\endbibitem

\bibitem[{Chib(1996)}]{Chib1996}
Chib, S. (1996).
\newblock \enquote{{Calculating posterior distributions and modal estimates in
  Markov mixture models}.}
\newblock {\em Journal of Econometrics\/}, 75(1): 79--97.
\endbibitem

\bibitem[{Chopin et~al.(2001)}]{Chopin2001}
Chopin, N. et~al. (2001).
\newblock {\em Sequential inference and state number determination for discrete
  state-space models through particle filtering\/}.
\newblock INSEE.
\endbibitem

\bibitem[{Churchill(1989)}]{Churchill1989}
Churchill, G.~a. (1989).
\newblock \enquote{{Stochastic models for heterogeneous DNA sequences}.}
\newblock {\em Bulletin of Mathematical Biology\/}, 51(1): 79--94.
\endbibitem

\bibitem[{Ding and Ou(2010)}]{Ding2010}
Ding, N. and Ou, Z. (2010).
\newblock \enquote{{Variational nonparametric Bayesian Hidden Markov Model}.}
\newblock {\em 2010 IEEE International Conference on Acoustics, Speech and
  Signal Processing\/}, 2098--2101.
\endbibitem

\bibitem[{{}Douc and {}Matias(2001)}]{DoMa01}
{}Douc, R. and {}Matias, C. (2001).
\newblock \enquote{{Asymptotics of the maximum likelihood estimator for general
  hidden Markov models}.}
\newblock {\em Bernouilli\/}, 7: 381--420.
\endbibitem

\bibitem[{{}Douc and {}Moulines(2012)}]{douc:moulines:12}
{}Douc, R. and {}Moulines, E. (2012).
\newblock \enquote{{Asymptotic properties of the maximum likelihood estimation
  in misspecified hidden Markov models}.}
\newblock 40(5): 2697--2732.
\endbibitem

\bibitem[{Douc et~al.(2004)Douc, Moulines, and Ryd\'{e}n}]{Ryden2004}
Douc, R., Moulines, E., and Ryd\'{e}n, T. (2004).
\newblock \enquote{{Asymptotic properties of the maximum likelihood estimator
  in autoregressive models with Markov regime}.}
\newblock {\em Annals of Statistics\/}, 32(5): 2254--2304.
\endbibitem

\bibitem[{Fox et~al.(2008)Fox, Sudderth, Jordan, and Willsky}]{Fox2008}
Fox, E.~B., Sudderth, E.~B., Jordan, M.~I., and Willsky, A.~S. (2008).
\newblock \enquote{{An HDP-HMM for systems with state persistence}.}
\newblock {\em Proceedings of the 25th international conference on Machine
  learning - ICML '08\/}, 312--319.
\endbibitem

\bibitem[{Friel and Pettitt(2008)}]{Friel2008}
Friel, N. and Pettitt, a.~N. (2008).
\newblock \enquote{{Marginal likelihood estimation via power posteriors}.}
\newblock {\em Journal of the Royal Statistical Society. Series B: Statistical
  Methodology\/}, 70(3): 589--607.
\endbibitem

\bibitem[{Fr\"{u}hwirth-Schnatter(2008)}]{Fruhwirth-Schnatter2006}
Fr\"{u}hwirth-Schnatter, S. (2008).
\newblock {\em {Finite Mixture and Markov Switching Models}\/}.
\newblock Mcmc. Springer, 1 edition.
\endbibitem

\bibitem[{Gassiat(2002)}]{Gassiat2002}
Gassiat, E. (2002).
\newblock \enquote{{Likelihood ratio inequalities with applications to various
  mixtures}.}
\newblock {\em Annales de l'Institut Henri Poincare (B) Probability and
  Statistics\/}, 38(6): 897--906.
\endbibitem

\bibitem[{Gassiat and Boucheron(2003)}]{Gassiat2003}
Gassiat, E. and Boucheron, S. (2003).
\newblock \enquote{{Optimal Error Exponents in Hidden Markov Models Order
  Estimation}.}
\newblock {\em Information Theory, IEEE \ldots\/}, 49(4): 964--980.
\endbibitem

\bibitem[{Gassiat and Keribin(2000)}]{Gassiat2000}
Gassiat, E. and Keribin, C. (2000).
\newblock \enquote{{The likelihood ratio test for the number of components in a
  mixture with Markov regime}.}
\newblock {\em ESAIM: Probability and Statistics\/}.
\endbibitem

\bibitem[{Gassiat and Rousseau(2012)}]{Gassiat2012}
Gassiat, E. and Rousseau, J. (2012).
\newblock \enquote{{About the posterior distribution in hidden markov models
  with unknown number of states}.}
\newblock {\em arXiv preprint arXiv:1207.2064\/}, (1997): 1--28.
\endbibitem

\bibitem[{Gelfand(1996)}]{Gelfand1996}
Gelfand, a. (1996).
\newblock \enquote{{Model Determination Using Sampling Based Methods}.}
\newblock {\em Markov Chain Monte Carlo in Practice\/}, 145--161.
\endbibitem

\bibitem[{Gunst and Shcherbakova(2009)}]{Gunst2009}
Gunst, M. C.~M. and Shcherbakova, O. (2009).
\newblock \enquote{{Asymptotic behavior of Bayes estimators for hidden Markov
  models with application to ion channels}.}
\newblock {\em Mathematical Methods of Statistics\/}, 17(4): 342--356.
\endbibitem

\bibitem[{Hamilton(1989)}]{Hamilton1989}
Hamilton, J.~D. (1989).
\newblock \enquote{{A new approach to the economic analysis of nonstationary
  time series}.}
\newblock {\em Econometrica\/}, 57(2): 357--384.
\endbibitem

\bibitem[{Han and Carlin(2001)}]{Han2001}
Han, C. and Carlin, B.~P. (2001).
\newblock \enquote{{Markov Chain Monte Carlo Methods for Computing Bayes
  Factors}.}
\newblock {\em Journal of the American Statistical Association\/}, 96(455):
  1122--1132.
\endbibitem

\bibitem[{McGrory and Titterington(2009)}]{McGrory2009}
McGrory, C.~a. and Titterington, D.~M. (2009).
\newblock \enquote{{Variational Bayesian Analysis for Hidden Markov Models}.}
\newblock {\em Australian \& New Zealand Journal of Statistics\/}, 51(2):
  227--244.
\endbibitem

\bibitem[{Nagaraja(2006)}]{Cappe2005}
Nagaraja, H.~N. (2006).
\newblock {\em {Inference in Hidden Markov Models}\/}, volume~48.
\endbibitem

\bibitem[{Nur et~al.(2009)Nur, Allingham, Rousseau, Mengersen, and
  McVinish}]{Nur2009}
Nur, D., Allingham, D., Rousseau, J., Mengersen, K.~L., and McVinish, R.
  (2009).
\newblock \enquote{{Bayesian hidden Markov model for DNA sequence segmentation:
  A prior sensitivity analysis}.}
\newblock {\em Computational Statistics \& Data Analysis\/}, 53(5): 1873--1882.
\endbibitem

\bibitem[{Richardson and Green(1997)}]{Richardson1997}
Richardson, S. and Green, P.~J. (1997).
\newblock \enquote{{On Bayesian analysis of mixtures with an unknown number of
  components}.}
\newblock {\em Journal of the Royal Statistical Society B\/}, 59(4): 731--792.
\endbibitem

\bibitem[{Rousseau and Mengersen(2011)}]{Rousseau2011}
Rousseau, J. and Mengersen, K. (2011).
\newblock \enquote{{Asymptotic behaviour of the posterior distribution in
  overfitted mixture models}.}
\newblock {\em Journal of the Royal Statistical Society. Series B: Statistical
  Methodology\/}, 73(5): 689--710.
\endbibitem

\bibitem[{Scott(2002)}]{Scott2002}
Scott, S.~L. (2002).
\newblock \enquote{{Bayesian methods for hidden Markov models: Recursive
  computing in the 21st century}.}
\newblock {\em Journal of the American Statistical Association\/}, 97(457):
  337--351.
\endbibitem

\bibitem[{Taylor et~al.(2012)Taylor, Tanner, and Wong}]{Tanner1987}
Taylor, P., Tanner, M.~a., and Wong, W.~H. (2012).
\newblock \enquote{{Augmentation The Calculation of Posterior Distributions by
  Data Augmentation}.}
\newblock {\em Journal of the American statistical Association\/}, 82(September
  2013): 37--41.
\endbibitem

\bibitem[{Teh et~al.(2006)Teh, Jordan, Beal, and Blei}]{Teh2006}
Teh, Y.~W., Jordan, M.~I., Beal, M.~J., and Blei, D.~M. (2006).
\newblock \enquote{{Hierarchical Dirichlet Processes}.}
\newblock {\em Journal of the American Statistical Association\/}, 101(476):
  1566--1581.
\endbibitem

\bibitem[{van Havre et~al.(2015)van Havre, White, Rousseau, and
  Mengersen}]{vanHavre2015}
van Havre, Z., White, N., Rousseau, J., and Mengersen, K. (2015).
\newblock \enquote{{Overfitting Bayesian Mixture Models with an Unknown Number
  of Components}.}
\newblock {\em Plos One\/}, 10(7): e0131739.
\endbibitem

\end{thebibliography}

\newpage
\appendix

\section{Proof of Theorem \ref{th:postempty}}	\label{Appendix1}

		As explained in Section \ref{Section1}, the proof is based on the following two Lemmas
		\begin{lem}\label{lem:LANSn}
		Under the hypotheses of Theorem \ref{th:postempty} (apart from assumption \textbf{[A3]} which is not needed),
		for all $\epsilon >0$, there exists $C_\epsilon >0$ such that for all $\theta \in S_n$
		\begin{equation*}
		\mathbb P^* \left( \ell_n(\theta, x_1=1) - \ell_n(\theta^*, x_1=1) < -C_\epsilon \right) < \epsilon
		\end{equation*}

		\end{lem}

		\begin{proof}[Proof of Lemma \ref{lem:LANSn}]
		      Throughout the proof $C$ is some generic constant whose value may vary.

		      The proof consists in proving considering a second order Taylor expansion of $\ell_n(\theta, x_1=1)$ around $\theta^*$ along the direction indicated by $S_n$, namely $q_{i,j}$ goes to $q_{i,j}^*$ for $i,j\leq K^*$ (we take $j=1$ as the reference column),  $q_{i,j} $ goes to $0$ for $j \geq k^*+1$ and $\gamma_i $ goes to $\gamma_i^*$ for all $i \leq K^*$. Hence Lemma \ref{lem:LANSn}, boils down to proving the gradient associated to these derivatives, computed at $\theta^*$, namely $\nabla \ell_n(\theta^*, x_1=1)$, is asymptotically Gaussian when normalized by $1/\sqrt{n}$ (it is enough to prove that it is $O_p(1)$) and that the second order derivatives computed at a random point in $(\theta, \theta^*)$, namely $D^2\ell_n(\bar \theta, x_1=1) $, is bounded from below by $-C $ times the identity matrix in $\R^{dK^*+ K^*(K^*-1)+ \underline \alpha K ( K-K^*)}$.

		      We write for all $\theta \in S_n$
		      \begin{equation*}
		      \ell_n(\theta, x_1=1) - \ell_n(\theta^*, x_1= 1)  = (\theta- \theta^*) \nabla \ell_n(\theta^*, x_1=1)+ \frac{ 1 }{ 2 } ( \theta - \theta^*)^tD^2 \ell_n(\bar \theta) (\theta - \theta^*) ,
		      \end{equation*}
		      where $\bar \theta \in \{ a \theta +(1-a)\theta^*, a\in (0,1)\}$ and
		      $$\theta^* = (\gamma_1^*, \cdots, \gamma_{K^*}^*, \gamma_{K^*+1}, \cdots, \gamma_K, q_{i,j}^*, i,j\leq K^*, q_{i,j}=0, \, j >K^*)$$
		      and $q_{i,j}>0$ and fixed for $i>K^*$ and $j\leq K_0$. In other words the Taylor expansion only concerns the coordinates $\gamma_i^*, i\leq K^*$, $q_{i,j}, $ $i,j\leq K^*$ and $q_{i,j}$, $j>K^*$.

		      We first look at the gradient. Note that, when computed at $\theta^*$, the first order derivatives associated to $q_{i,j}$, $i,j\leq K^*$ and $\gamma_i$, $i\leq K^*$ are the same as in the well specified case and we can apply directly Theorem 2 of  \citet{Ryden2004}, so that
		       $$n^{-1/2} \nabla_{\theta_l} \ell_n(\theta^*, X_1=x) \rightarrow \mathcal N(0, I(\theta^*)), $$
		       in distribution under $\mathbb P^*$.
		       for some finite matrix $I(\theta^*)$ for all $\theta_l \in \{q_{i,j},\,  i,j\leq K^*,\,  j\geq 2, \,\gamma_i \, i \leq K^* \}$.
		       We thus need only study
		       $\nabla_{q_{i,j}}\ell_n(\theta^*, x_1=1)$, for $i \leq K $ and $j\geq K^*+1$. The control of  $\nabla_{q_{i,j}}\ell_n(\theta^*, x_1=1)$ for $j \geq K^*+1$ follows the same lines as  the proof of Proposition 5 of \citet{DoMa01}.
		       First note that for all $s\geq K_0+1$
		       $$\nabla_{q_{r,s}} \ell_n(\theta^*, x_1=x)  = \sum_{i=1}^n \nabla_{\theta_l} \log f(y_i| y_{1:i-1}, x_1= 1; \,  \theta^*) $$
		       and that
		       $$ \nabla_{q_{r,s}}\log f(y_i| y_{1:i-1}, x_1= 1; \,  \theta^*) = \sum_{j=1}^K \frac{ \nabla_{q_{r,s}} p_{\theta^*}( x_{i}=j|y_{1:i-1},x_1=1)g_{\gamma_j}(y_i) }{  f(y_i| y_{1:i-1}, x_1= 1; \,  \theta^*)} $$
		       with
		       \begin{equation*}
		       \begin{split}
		       \nabla_{q_{r,s}} p_{\theta^*}( x_{i}=j|y_{1:i-1},x_1=1) &= \frac{ \sum_{l=1}^{K^*} q_{l,j}^*g_{\gamma_l^*}(y_{i-1}) \nabla_{q_{r,s}} p_{\theta^*}( x_{i-1}=l|y_{1:i-2})}{ f(y_{i-1}| y_{1:i-2}, x_1= 1; \,  \theta^*)} \\
		       & - p_{\theta^*}( x_{i}=j|y_{1:i-1},x_1=1)\frac{ \sum_{l=1}^{K^*} g_{\gamma_l^*}(y_{i-1}) \nabla_{q_{r,s}} p_{\theta^*}( x_{i-1}=l|y_{1:i-2})}{ f(y_{i-1}| y_{1:i-2}, x_1= 1; \,  \theta^*)}\end{split}
		       \end{equation*}
		       and for all $j >K^*$, $ \nabla_{q_{r,s}} p_{\theta^*}( x_{i}=j|y_{1:i-1},x_1=1)= 0$, together with $  p_{\theta^*}( x_{i}=j|y_{1:i-1},x_1=1)$ for all $i$.
		       In the sequel, to ease notations we omit $x_1=1$ in the notations, but every conditional distributions has initial value $x_1=1$.
		       Thus eq (35) in \citet{DoMa01} is verified with
		        $$ a_{\theta^*}( y, f)(u,v)  = \left( q^*_{u,v} - \frac{ \sum_i q^*_{i,v}g_{\gamma_i}(y)f(i)}{ \sum_i g_{\gamma_i}(y)f(i)} \right)\frac{ g_{\gamma_u}(y)}{\sum_i g_{\gamma_i}(y)f(i)}$$
		       \begin{equation*}
		       \begin{split}
		       U_{\theta^*,q_{r,s}}( y,f) (v)& = \frac{ g_{\gamma_r}(y)f( r ) \left( \1_s(v) - \1_1(v) \right)}{ \sum_i g_{\gamma_i}(y)f(i)}, \quad \mbox{if } \quad r \leq K^*\\
		       U_{\theta^*,q_{r,s}}( y,f) (v)& = 0, \quad \mbox{if } \quad r >K^*.
		       \end{split}
		       \end{equation*}
		       We restrict ourselves to $f \in \mathbf S^+_{\underline q} = \{ f  : \{1, \cdots, K^*\} \rightarrow [0,1], \, \sum_u f(u) =1, f\geq \underline q\}$. On this set
		        $$ \|a_{\theta^*}( y, f) - a_{\theta^*}( y, f')\|_1 \leq \frac{ \|f - f'\|_2 }{ \underline q^2}, \quad  \| U_{\theta^*,q_{r,s}}( y,f) -  U_{\theta^*,q_{r,s}}( y,f')\|_1 \leq \frac{ 3 }{ \underline q} \|f - f'\|_1  .$$
		      Using the same computations as in the proof of Proposition 5 of \citet{DoMa01}, we obtain that
		      $$ E_{\theta^*}\left[ \frac{1 }{n}\left( \nabla_{q_{r,s}} \ell_n(\theta^*, x_1=1) \right)^2 \right] = O(1) .$$

		      We now study the second order derivatives.
		       We use  the expression (22) in \citet{Ryden2004} :
		      \begin{equation}\label{22:DMR}
		      D^2 \ell_n( \theta; x_1= 1) = \mathbb E_\theta\left[ \sum_{i=1}^n \varphi( \theta, z_{i-1:i} )|y_{1:n}, x_1=1\right]+
		      var_\theta\left[ \sum_{i=1}^n \phi( \theta, z_{i-1:i} )|y_{1:n}, x_1=1\right]
		      \end{equation}
		      where $z_i = (x_i, y_i) $ and
		      $$ \varphi(\theta; z_{i-1:i}) =  D^2 \log ( q_{x_{i-1},x_i}g_{\gamma_{x_i}}(y_i) ), \quad
		       \phi(\theta; z_{i-1:i}) =  \nabla \log ( q_{x_{i-1},x_i}g_{\gamma_{x_i}}(y_i) ).$$
		      Note that for $l \leq K_0$,
		      \begin{equation*}
		       \nabla_{\gamma_l} \log ( q_{x_{i-1},x_i}g_{\gamma_{x_i}}(y_i) ) = \nabla \log g_{\gamma_l}(y_i) \1_{x_i=l}, \quad
		        \nabla_{q_{r,s}} \log ( q_{x_{i-1},x_i}g_{\gamma_{x_i}}(y_i) ) = \1_{x_{i-1}=r}\left(  \frac{ \1_{x_i=s}}{q_{r,s}} - \frac{\1_{x_i=1}}{q_{r,1}}\right)
		      \end{equation*}
		      So that
		      \begin{equation*}
		      \begin{split}
		      D^2_{\gamma_l, \gamma_j} \log ( q_{x_{i-1},x_i}g_{\gamma_{x_i}}(y_i) ) &=D^2 \log g_{\gamma_l}(y_i) \1_{x_i=l=j}, \quad D^2_{\gamma_l, q_{r,s}} \log ( q_{x_{i-1},x_i}g_{\gamma_{x_i}}(y_i) ) =0\\
		      D^2_{q_{r,s},q_{r's'}} \log ( q_{x_{i-1},x_i}g_{\gamma_{x_i}}(y_i) ) &=- \1_{x_{i-1}=r=r'}\left(  \frac{ \1_{x_i=s=s'}}{q_{r,s}^2} + \frac{\1_{x_i=1}}{q_{r,1}^2}\right)
		      \end{split}
		      \end{equation*}
		      Note hat we need only bound from below $(\theta - \theta^*)^tD^2 \ell_n(\bar\theta, x_1=1)(\theta-\theta^*)$ and that for all $s>K^*$ $q_{r,s}>q_{r,s}^*$. Hence  for derivatives associated to $\gamma_l$, it is enough to bound from above
		       $$\frac{ 1}{ n }  \mathbb E_\theta\left[ \sum_{i=1}^n \varphi( \theta, z_{i-1:i} )|y_{1:n}, x_1=1\right]  \leq \frac{ 1 }{ n } \sum_{i=1}^n |D^2 \log g_{\gamma_l}(y_i)| = O_p(1) $$
		       uniformly for $\theta \in S_n$.

		      For $\varphi $ associated to derivatives with respect to $q_{r,s}$ with $s\leq K^*$,
		       $\left|\mathbb E_\theta\left[ \sum_{i=1}^n \varphi( \theta, z_{i-1:i} )|y_{1:n}, x_1=1\right] \right|$
		        is bounded by
		         or $ \left( \frac{ 1 }{ q_{r,s}^2} + \frac{ 1 }{ q_{r,1}^2 } \right)$
		      and
		      \begin{equation*}
		      \frac{1}{n }\sup_{\theta \in S_n }\left|  \mathbb E_\theta\left[ \sum_{i=1}^n \varphi( \theta, z_{i-1:i} )|y_{1:n}, x_1=1\right]\right| = O_p(1)
		      \end{equation*}
		      To handle the variance term, note that in the case of $\gamma_l$ derivatives the variance terms are positive so we can bound from below by 0.
		      For derivatives associated to $q_{r,s}$ $q_{rs'}$, with either $s\neq s'$ or $s = s'\leq K^*$, for instance in the case where $s\neq s'$
		      \begin{equation}\label{Aqrss'}
		      \begin{split}
		      A_{q_{r,s}q_{rs'}} &= var_\theta\left[ \sum_{i=1}^n \phi( \theta, z_{i-1:i} )|y_{1:n}, x_1=1\right]\\
		       &\leq \frac{ n }{q_{r,1}^2} + \frac{ 2 }{q_{r,1}^2}  \sum_i \sum_{k=2}^{n-i}  p_\theta\left( x_{i-1}=r , x_i =1 |y_{1:n}, x_1=1\right)\times \\
		      & \quad \quad\left| p_\theta\left( x_{i+k-1}=r, x_{i+k}=1 |y_{1:n}, x_i = 1\right) -  p_\theta\left( x_{i+k-1}=r, x_{i+k} =1 |y_{1:n}, x_1=1\right)\right|.
		      \end{split}
		      \end{equation}
		      In the case of $s=s'\geq K^*+1$, then
		      \begin{equation}\label{Aqrs}
		      \begin{split}
		      A_{q_{r,s}q_{r,s}} &= \mathbb E_\theta  \left[ \sum_{i=1}^n \varphi( \theta, z_{i-1:i} )|y_{1:n}, x_1=1\right] + var_\theta\left[ \sum_{i=1}^n \phi( \theta, z_{i-1:i} )|y_{1:n}, x_1=1\right]\\
		      &= -\frac{ 1 }{ q_{r,s}^2 }  \mathbb E_\theta\left[ \sum_{i=1}^n \1_{x_{i-1} = r , x_i = s }|y_{1:n}, x_1=1\right] +\frac{ 1 }{ q_{r,s}^2 }  var_\theta\left[ \sum_{i=1}^n \1_{x_{i-1} = r , x_i = s }|y_{1:n}, x_1=1\right]\\
		      & -\frac{ 1 }{ q_{r,1}^2 }  \mathbb E_\theta\left[ \sum_{i=1}^n \1_{x_{i-1} = r , x_i = 1 }|y_{1:n}, x_1=1\right] +\frac{ 1 }{ q_{r,1}^2 }  var_\theta\left[ \sum_{i=1}^n \1_{x_{i-1} = r , x_i = 1 }|y_{1:n}, x_1=1\right] \\
		      & + \frac{ 2}{ q_{r,s}q_{r,1} } cov_\theta\left[  \sum_{i=1}^n \1_{x_{i-1} = r , x_i = s },  \sum_{i=1}^n \1_{x_{i-1} = r , x_i = 1 } | y_{1:n}, x_1=1 \right] \\
		      &\geq  - \frac{ 1 }{ q_{r,s}^2 }\sum_j p_\theta (x_{j-1} = r , x_j = s |y_{1:n}, x_1= 1)^2\\
		      &
		      + \frac{ 1 }{ q_{r,s}^2 } \sum_{i<j}\mathbb E_\theta\left[  \1_{x_{i-1} = r , x_i = s }\left( \1_{x_{j-1} = r , x_j = s } - p_\theta (x_{j-1} = r , x_j = s |y_{1:n}, x_1= 1) \right) |y_{1:n}, x_1=1\right] \\
		      & +\frac{ 2}{ q_{r,s}q_{r,1} } \sum_{i\neq j} \mathbb E_\theta\left[ ( \1_{x_{i-1} = r , x_i = s } - p_\theta (x_{i-1} = r , x_i = s |y_{1:n}, x_1= 1) ) \1_{x_{j-1} = r , x_j = 1 }|y_{1:n}, x_1=1\right].
		      \end{split}
		      \end{equation}
		      Note that
		      \begin{equation}\label{upboundprs}
		      \begin{split}
		      p_\theta (x_{j-1} = r , x_j = s |y_{1:n}, x_1= 1) &=\mathbb E_\theta \left( p_\theta( x_j=s |y_s, x_{j-1}=r,x_{j+1}) | y_{1:n}, x_1= 1\right)\\
		      &= \mathbb E_\theta \left( \frac{ q_{r,s}g_{\gamma_s}(y_j)q_{s,x_{j+1}} }{ \sum_lq_{rl}g_{\gamma_l}(y_j)q_{lx_{j+1}} } | y_{1:n}, x_1= 1 \right)\\
		      &\leq q_{r,s} \frac{ g_{\gamma_s}(y_j)}{ \underline q\max_{l\leq K^*} g_{\gamma_l}(y_j)} \leq \frac{ q_{r,s} }{ \underline q}  \| \frac{ g_{\gamma_s} }{ \max_{l\leq K^*} g_{\gamma_l} }\|_\infty
		      \end{split}
		      \end{equation}
		      and the first term of the right hand side of the above inequality is bounded from below, using assumption \textbf{A2}, by
		      $$ - \frac{ nC }{ \underline q^2 } , $$
		      for some $C>0$.
		      We thus need to control
		      \begin{equation*}
		      \Delta_{k,n}(x) =   p_\theta\left( x_{k}=r |y_{1:n}, x_1 = 1\right) -  p_\theta\left( x_{k}=r |y_{1:n}, x_1=x\right).
		      \end{equation*}
		      This has been done in many contexts, however to the best of our knowledge, this has not been done in cases where the matrix $Q$ has entries that satisfy
		       $ q_{i,j} \geq c$ for all $i $ and all $j\leq K^*$ and $q_{i,j}\in (1/2, 1)/\sqrt{n}$ for all $j > K^*$.
		      We can write
		       \begin{equation*}
		       \begin{split}
		      \Delta_{k,n}(x) &=   \mathbb E_{\theta}\left[\left. p_\theta\left( x_{k}=r |y_{1:n},x_{k+1}, x_1 = 1\right) -  p_\theta\left( x_{k}=r |y_{1:n}, x_{k+1},x_1=x\right)\right| y_{1:n}, x_1=x\right] \\
		      &:=  \mathbb E_{\theta}\left[\Delta_{k,n}(x, x_{k+1})|  y_{1:n}, x_1=x\right]
		      \end{split}
		      \end{equation*}
		      and note that
		      \begin{equation*}
		      \begin{split}
		      \Delta_{k,n}(x, x_{k+1})& = p_\theta\left( x_{k}=r |y_{1:k},x_{k+1}, x_1 = 1\right) -  p_\theta\left( x_{k}=r |y_{1:k}, x_{k+1},x_1=x\right) \\
		      &= \frac{ \mathbf L^\theta<y_1^k>(1,q_{.,x_{k+1}})}{ \mathbf L^\theta<y_1^k>(1,1)} - \frac{ \mathbf L^\theta<y_1^k>(x,q_{.,x_{k+1}})}{ \mathbf L^\theta<y_1^k>(x,1)}\\
		      &= \frac{ \mathbf L^\theta<y_1^k>\otimes \mathbf L^\theta<y_1^k> (1\otimes  x, q_{.,x_{k+1}}\otimes 1 - 1 \otimes q_{.,x_{k+1}}) }{ \mathbf L^\theta<y_1^k>\otimes \mathbf L^\theta<y_1^k> (1\otimes  x, 1\otimes 1)}
		      \end{split}
		      \end{equation*}
		      where
		      $$\mathbf L^\theta<y_1^k>(x,f ) = \sum_{x_2, \cdots, x_k}\prod_{i=1}^{k-1} g_{\gamma_{x_i}}(y_i)q(x_i,x_{i+1})f(x_{k+1}), \quad x_1 = x$$
		      as in \citet{douc:moulines:12}. Note that assumption \textbf{A1} of \citet{douc:moulines:12} is satisfied with $C = \{1, \cdots, K^*\}$, their set $\mathbf K = \R^d$ and $D = \{1, \cdots, K\}$. Hence, using Proposition 5 of \citet{douc:moulines:12}, we obtain that there exists $1 >\rho>0$ and $c>0$ (independent of $\theta$) such that for all $\eta >0$, when $n$ is large enough
		      \begin{equation*}
		      \begin{split}
		      & \left|\mathbf L^\theta<y_1^k>\otimes \mathbf L^\theta<y_1^k> (1\otimes  x, q_{.,x_{k+1}}\otimes 1- 1 \otimes q_{.,x_{k+1}})\right|\\
		      & \leq \rho^{ck} \left[ \mathbf L^\theta<y_1^k>(1,q_{.,x_{k+1}})\mathbf L^\theta<y_1^k>(x,1)+  \mathbf L^\theta<y_1^k>(x,q_{.,x_{k+1}}) \mathbf L^\theta<y_1^k>(1,1) \right]\\
		      & +2\eta^{k} \prod_{i=2}^k\| \mathbf L^\theta<y_i>(.,1 )\|_\infty^2 \|q_{.,x_{k+1}}\|_\infty L^\theta<y_1>(1,1 )L^\theta<y_1>(x,1 )\\
		      &:= I_1 + I_2
		      \end{split}
		      \end{equation*}
		      The first term leads to an upper bound of order $\rho^{cn}$ in $\Delta_{k,n}(x, x_{k+1})$. The second is controlled in th following way:
		      Since $\| \mathbf L^\theta<y_i>(.,1 )\|_\infty = \max_u \sum_j q_{u,j}g_{\gamma_u}(y_i) = \max_u g_{\gamma_u}(y_i)$ and
		      \begin{equation*}
		      \begin{split}
		       \mathbf L^\theta<y_1^k>(x, 1) &= \sum_{x_2,\cdots, x_k} \prod_{i=2}^kq_{x_{i-1},x_i}g_{\gamma_{x_{i-1}}}(y_{i-1})g_{\gamma_{x_k}}(y_k), \quad x_1 = x\\
		       &\geq \underline q^k g_{\gamma_x}(y_1)\prod_{i=2}^k \max_{l\leq K^*}g_{\gamma_{l}}(y_{i}),
		      \end{split}
		      \end{equation*}
		      we obtain
		      \begin{equation*}
		      \frac{I_2 }{ \mathbf L^\theta<y_1^k>\otimes \mathbf L^\theta<y_1^k> (1\otimes  x, 1\otimes 1)} \leq 2\eta^{k} \underline q^{-k} \prod_{i=2}^k \left(\frac{ \max_u g_{\gamma_u}(y_i) }{ \max_{l\leq K_0}g_{\gamma_{l}}(y_{i}) }\right)^2  \leq 2\frac{\eta^k}{ (V\underline q )^k}
		      \end{equation*}
		      for some $V>0$ under assumption \textbf{A2} and by choosing $\eta$ small enough we finally obtain that there exists $\rho\in (0,1)$ such that
		      \begin{equation}\label{Delta_kn}
		      \Delta_{k,n}(x, x_{k+1})\leq \rho^k .
		      \end{equation}
		      Combining with \eqref{Aqrss'} and the computations of \eqref{upboundprs} we can bound
		      \begin{equation*}
		      A_{q_{r,s}q_{rs'}} \leq \frac{ 2n}{ \underline q^2} (1 + \sum_k \rho^k )\lesssim n.
		      \end{equation*}
		      Similarly, the second term of the right hand side of \eqref{Aqrs} can be bounded by
		      \begin{equation*}
		      \begin{split}
		      & \frac{ 1 }{ q_{r,s}^2 } \sum_{i<j}\mathbb E_\theta\left[  \1_{x_{i-1} = r , x_i = s }\left( \1_{x_{j-1} = r , x_j = s } - p_\theta (x_{j-1} = r , x_j = s |y_{1:n}, x_1= 1) \right) |y_{1:n}, x_1=1\right] \\
		      &\leq  \frac{ 1 }{ q_{r,s}^2 }\sum_i p_\theta\left( x_i=s| x_{i-1}=r , y_{i-1:n}\right) \sum_{k=2}^{n-i}   p_\theta\left( x_{i+k}=s| x_{i+k-1}=r , y_{i+k-1:n}\right) \times \\
		      & \quad \left[  p_\theta\left( x_{i+k-1}=r | y_{i:n}, x_i= r\right)-  p_\theta\left( x_{i+k-1}=r | y_{1:n}, x_1= 1\right)\right] \\
		      &\leq \frac{ C }{\underline q^2 } \sum_i \sum_{k=2}^{n-i } \rho^{k-1} \lesssim n \underline q^{-2}
		      \end{split}
		      \end{equation*}
		      and similarly with the third term of \eqref{Aqrs} . Finally we obtain that with probability going to 1, for all $\theta \in S_n$,
		      $$\ell_n(\theta, x_1=1) - \ell_n(\theta^*,x_1=1) \geq - Cn $$
		      for some constant $C$ large enough.

		      \end{proof}
		      \begin{lem}\label{lem:upb:Bn}
		      Under the hypotheses of Theorem \ref{th:postempty},
		      Let $B_n = \{ \theta \in A_n; \min_{\sigma \in \mathcal S_K } \sum_{i>K^*}\mu_Q(\sigma(i) > M_n u_n\}$ with $M_n$ any sequence going to infinity, then
		      $$\pi(B_n) \lesssim  u_n^{K^*(d+K^*-1)+ \underline{\alpha}(K^*+1)(K-K^*-1)+ d/2} $$
		      \end{lem}

		      \begin{proof}[Proof of Lemma \ref{lem:upb:Bn}]
		      We set
		      $$ A_n=\{ || f_{2}(.|\theta)-f_{2}(.|\theta^*)||_1 \lesssim u_n \}, \quad u_n=e_n^{-1} \sqrt{\log n} \left(\frac{1}{\sqrt{n}} \right)^{1-A} $$
		      where $A = [K^*( K^*-1+d)+ \underline \alpha K(K-K^*)]/(p\bar \alpha + (K-p) \underline \alpha )$.

		      We are now interested in considering the subset of $A_n$ corresponding to not emptying completely the extra components. We use (26) of \citet{Gassiat2012}.  Hence there exists  $\epsilon >0$ such that  for all $i \leq K^*$, defining $A(i) = \{ j; \|\gamma_j - \gamma_i^*\|\leq \epsilon\}$ and writing $\Gamma^*=\{\gamma_i^*, i\leq K^*\}$, we have
		      \begin{align}\label{eq:3_1}
		      u_n &\gtrsim \sum_{j:||\gamma_1-\Gamma^*||_1 >\epsilon } \mu_Q(j) +\sum_{i_1, i_2}|\sum_{j_1\in A(i_1), j_2\in A(i_2)} \mu_Q(j_1) q_{j_1, j_2} - \mu^*(i_1)q^*_{i_1,i_2}| \\
		      &+ \sum_{i_1, i_2} |\sum_{j_1 \in A(i_1), j_2 \in A(i_2)} \mu_Q(j_1) q_{j_1, j_2}(\gamma_{j_1}, \gamma_{j_2})^T - \mu^*(i_1)q^*_{i_1, i_2}(\gamma_{i_1}^*, \gamma^*_{i_2})^T |  \nonumber     \\
		      &+ \sum_{i_1, i_2} \sum_{j_1 \in A(i_1), j_2 \in A(i_2)}  \mu_Q(j_1)q_{j)1, j_2} ( ||\gamma_{j_1} - \gamma_{i_1}^* ||^2 +||\gamma_{j_2} - \gamma_{i_2}||^2).\nonumber
		      \end{align}
		      We can thus partition $A_n$ into all the possibilities of constructing $(A(i), i \leq K^*)$.
		      For each of these partition, we express the event corresponding with merging at least two components. To begin with, to simplify the presentation of the calculations and without loss of generality we set $K_t$, the  cardinal of $\cup_{i\leq K_0}A(i)$ so that if $\theta \in A_n$, $K_t \geq K^*+1$. To simplify the notations we assume that  $\cup_{i\leq K^*}A(i) = \{1,\dots, K_t\}$, and that the indices belonging to $A(i)$ are all smaller than those belonging to $A(i+1)$ for all $K^* \geq i \geq 1$.

		      \underline{Case where $K_t=K^*+1$}

		       This means that component $\{K^*+1\}$ (up to a permutation) belongs to one of the  $A(i)$'s, which we set to be $K^*$. Thus, $A(i)=\{i\}$ for $i\leq K^*-1$ and $A(K^*)=\{K^*, K^*+1\}$.

		      Then, \eqref{eq:3_1} implies, for all $i_1, i_2 \leq K^*-1$,
		      $$|\mu_Q(i_1)-\mu^*(i_1)|\lesssim u_n  , \quad |q_{i_1, i_2}-q^*_{i_1, i_2}|\lesssim u_n, \quad \sum_{j>K^*+1} \mu_Q(j)\lesssim u_n$$.
		      Since we can assume without loss of generality that $\mu_Q(K^*) \geq \mu_Q(K^*+1)$ (up to relabelling)
		      and since $|\mu_Q(K^*)+\mu_Q(K^*+1)- \mu^*(K^*)| \lesssim u_n$, using $\mu_QQ=\mu_Q$, this implies that
		      $$ \forall i \leq K^*, \forall j > K^*+1,\quad q_{i,j} \lesssim u_n, \quad q_{K^*+1, j} \lesssim \frac{u_n}{\mu_Q(K^*+1)}$$

		      We also have that for all $i\leq K^*-1$ $$q_{K^*,i}=(\mu^*(K^*)q^*_{K^*,i}-q_{K^*+1,i}\mu_Q(K^*+1))/(\mu^*(K^*)-\mu_Q(K^*+1)) \pm O(u_n), $$
		      and for all $i\leq K^*-1$
		      $$ ||\gamma_i -\gamma^*_i ||\lesssim u_n,$$
		      $$||\gamma_{K^*}+\mu_Q(K^*+1)\gamma_{K^*+1}-\mu^*(K^*)\gamma^*_{K^*}||\lesssim u_n$$
		      $$\mu_Q(K^*+1)||\gamma_{K^*+1}-\gamma^*_{K^*} ||^2\lesssim u_n .$$

		      Since, if in addition $\theta \in  B_n$, then  $\mu_Q(K^*+1)\geq v_n$.  We denote $\pi_{n,1}$ the prior mass of $B_n$ intersecting with the set of $\theta$'s satisfying the partition $A(i) = i$, for all $i \leq K^*-1$ and $A(K^*) = \{K^*, K^*+1\}$. Then
		      $$\pi_{n,1}\lesssim  u_n^{dK^*}u_n^{K^*(K^*-1)} u_n^{\underline{\alpha}(K^*+1)(K-K^*-1)}u_n^{d/2} \times J_n$$
		      with $C_n = \{ \sum_{i \leq K^*}q_{iK^*+1} \gtrsim  v_n(1 - q_{K^*+1,K^*+1})\}$ and
		      \begin{equation*}
		      \begin{split}
		      J_n &=  \int_{C_n} \prod_{i \leq K^*+1}\frac{  q_{i, K^*+1}^{\underline \alpha-1} }{ \left( \sum_{j \leq K^*}q_{i,K^*+1}\right)^{ d/2 +\underline \alpha(K-K^*-1) }}(1 - q_{K^*+1,K^*+1})^{ d/2 +\underline \alpha( K-K^*-1) + p\bar \alpha + (K^*-p)\underline \alpha -1 } dq_{iK^*+1} \\
		      &\lesssim \int_{\sum_j x_j \gtrsim v_n} \frac{ \prod_{i\leq K^*} x_i^{\alpha_i-1} }{ (\sum_j x_j)^{d/2 + \underline \alpha (K-K^*-1)}}dx_i \times
		      \int_0^1q^{\underline \alpha -1} (1-q)^{p\bar \alpha +2(K^*-p)\underline \alpha  -K^*-1}dq
		      \end{split}
		      \end{equation*}
		      Hence if $p\bar \alpha +2(K^*-p)\underline \alpha >K^*$,
		      \begin{equation*}
		      \begin{split}
		      J_n
		      &\lesssim \int_{\sum_j x_j \gtrsim v_n} \frac{ \prod_{i\leq K^*} x_i^{\alpha_i-1} }{ (\sum_j x_j)^{d/2 + \underline \alpha (K-K^*-1)}}dx_i \lesssim  v_n^{-d/2-\underline{\alpha}(K-2K^*-1)}
		      \end{split}
		      \end{equation*}
		      This leads to, with $B= K^*(d+K^*-1)+ \underline{\alpha}(K^*+1)(K-K^*-1)+ d/2$
		      \begin{align*}
		      \pi_{n,1} \lesssim & u_n^{B}v_n^{-d/2-\underline{\alpha}(K-2K^*-1)} = (\log n)^{B/2}n^{-(1-A)B/2} e_n^{-1}v_n^{-d/2-\underline{\alpha}(K-2K^*-1)},
		      \end{align*}
		      for any $e_n= o(1) $. Hence $\pi_{n,1} = e_n n^{-(K^* (K^*-1+d)+ \underline \alpha K( K-K^*))/2} $ if $(1-A)B>K^* (K^*-1+d)+ \underline \alpha K( K-K^*)$, under a suitable choice of $v_n$.
		      \begin{equation*}
		      \begin{split}
		      & (K^*(d+K^*-1)+ \underline{\alpha}(K^*+1)(K-K^*-1)+ d/2)(1-A) > K^* (K^*-1+d)+ \underline \alpha K( K-K^*)\\
		      \Leftrightarrow & d/2 - A(K^*(d+K^*-1)+ \underline{\alpha}(K^*+1)(K-K^*-1)+ d/2)> \underline \alpha[ (K-K^*)^2 -(K-2K^*-1)]\\
		      \Leftrightarrow & \frac{ (K^* (K^*-1+d)+ \underline \alpha K( K-K^*))(K^*(d+K^*-1)+ \underline{\alpha}(K^*+1)(K-K^*-1)+ d/2)  }{ p\bar \alpha + (K-p) \underline \alpha } \\
		      & \qquad  < d/2 - \underline \alpha[ (K-K^*)^2 -(K-2K^*-1)]
		      \end{split}
		      \end{equation*}
		       This is satisfied if $d/2 > \underline \alpha[ (K-K^*)^2 -(K-2K^*-1)]$ and
		       \begin{equation*}
		        p\bar \alpha + (K-p) \underline \alpha > \frac{ (K^* (K^*-1+d)+ \underline \alpha K( K-K^*))(K^*(d+K^*-1)+ \underline{\alpha}(K^*+1)(K-K^*-1)+ d/2)  }{ d/2 - \underline \alpha[ (K-K^*)^2 -(K-2K^*-1)]}.
		       \end{equation*}

		      Using similar, but more tedious computations we can prove that when $K_t > K^*+1$, the prior mass is  smaller than $ u_n^{B}v_n^{-d/2-\underline{\alpha}(K-2K^*-1)}$, which concludes the proof of Lemma \ref{lem:upb:Bn}.

		  \end{proof}

		\section*{Gibbs sampler for overfitted HMMs}\label{sec:AppendixHMM_Gibbs}
		A Gibbs sampler is set up on the augmented parameter space, $$p(X, \gamma, Q | Y) \propto p(Y|X,  \gamma ,Q )p(X|\gamma ,Q ) \pi(\gamma) \pi(Q)$$  as desribed by \citet{Fruhwirth-Schnatter2006}. The prior on the emission means is set to $\pi(\gamma)\sim \mathcal{N}(\gamma_0=\bar{Y}, \tau_0=100$, where $\bar{Y}$ is the observed sample mean. The prior on each row of $Q$ follows a Dirichlet distribution of the form $\mathcal{D}(\bar{\alpha}, \underline{\alpha}, \dots, \underline{\alpha})$.

		The Gibbs sampler is straightforward, keeping in mind the initial distribution of the hidden states is assumed to be equal to the ergodic distribution, resulting in a stationary Markov chain. As the rows of $Q$ are no longer independent \textit{a posteriori}, a Metropolis-Hastings step is incorporated into a standard the Gibbs sampler (note Step 2.a.i). $X^{{j}^m}_0$ denotes the estimated initial state of chain $j$ at iteration $m$.

		\paragraph{Gibbs sampler for overfitted HMMs}
		\begin{enumerate}
		\item Initialise:
		    \begin{enumerate}
		    \item Chose  $\underline{\alpha}$ and $K$, and set $\bar{\alpha}$ so that $\pi(q_{i.})\sim \mathcal{D}=\left(\bar{\alpha}, \underline{\alpha},\dots,  \underline{\alpha}\right)$
		    \item Choose a set of starting values for the allocations,   $X^0$
		    \end{enumerate}
		\item Step $m$, for each iteration $m=1,\dots, M$:
		    \begin{enumerate}
		    \item \textit{Gibbs Sampling}.
		        \begin{enumerate}
		        \item Generate transition matrix $Q^{m}$ given previous states from $p(Q^{m}| X^{m-1})$,
		            \begin{itemize}
		            \item Accept $Q^{m}$ with probability $min \left(1,\dfrac{\mu^{m}_{ X^{m}_0 }}{\mu^{m-1}_{ X^{m}_0}}\right)$
		            \end{itemize}
		        \item Generate $\theta_k^m$ from $p(\theta_k^m|Y, X^{m-1})$ for each $k=1,\dots, K$ (in our case, only the emissions means, $\gamma_j$, need be estimated here)
		        \item Generate $p(\textbf{X}|Q^{m},\gamma^{m}, Y)$ using the forward backward algorithm \citep{Cappe2005, Fruhwirth-Schnatter2006}.
		        \end{enumerate}
		    \end{enumerate}
		\end{enumerate}

\section*{Figures}

    \begin{figure}[htb]
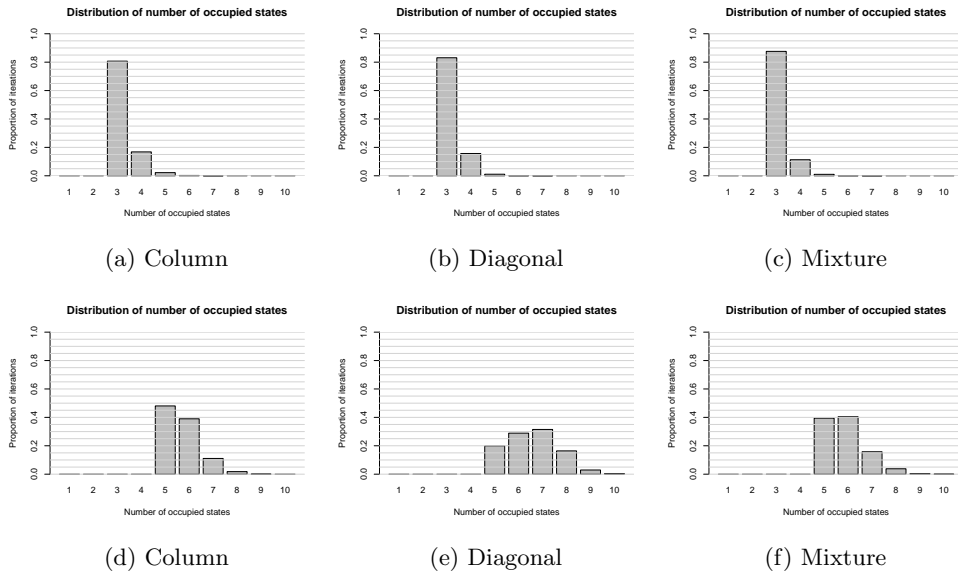

    \centering\begin{subfigure}{0.32\textwidth}
    \adjincludegraphics[trim={0 {.5\height} 0 0},clip,width=1\linewidth]{Fig4_1}
    \caption{Column}
    \end{subfigure}
      \begin{subfigure}{0.32\textwidth}
    \adjincludegraphics[trim={0 {.5\height} 0 0},clip,width=1\linewidth]{Fig4_2}
    \caption{Diagonal}
    \end{subfigure}
    \begin{subfigure}{0.32\textwidth}
    \adjincludegraphics[trim={0 {.5\height} 0 0},clip,width=1\linewidth]{Fig4_3}
    \caption{Mixture}
    \end{subfigure} \\
    \begin{subfigure}{0.32\textwidth}
    \adjincludegraphics[trim={0 {.5\height} 0 0},clip,width=1\linewidth]{Fig4_4}
    \caption{Column}
    \end{subfigure}
      \begin{subfigure}{0.32\textwidth}
    \adjincludegraphics[trim={0 {.5\height} 0 0},clip,width=1\linewidth]{Fig4_5}
    \caption{Diagonal}
    \end{subfigure}
    \begin{subfigure}{0.32\textwidth}
    \adjincludegraphics[trim={0 {.5\height} 0 0},clip,width=1\linewidth]{Fig4_6}
    \caption{Mixture}
    \end{subfigure} 
     \caption[Small sample demonstration: bivariate posterior densities]{Distribution of the number of occupied states produced from fitting a HMM with $K=10$ states to $n=100$ observations from a simulated HMM from Sim 2 and Sim 3. Each row includes a plot for each choice of prior. }  
     \label{fig:HMM_Figure_6}
\end{figure}

\end{document}